\documentclass[11pt, letterpaper]{article}
\usepackage[colorlinks=true,linkcolor=red,citecolor=blue]{hyperref}
\usepackage{fullpage,paralist}
\usepackage{theorem}
\usepackage{times}
\usepackage{booktabs}
\usepackage{stmaryrd}

\usepackage{verbatim}
\usepackage{amssymb}
\usepackage{amsmath}
\usepackage{ifthen}

\newtheorem{theorem}{Theorem}[section]

\newtheorem{remark}{Remark}
\newtheorem{proposition}{Proposition}
\newtheorem{corollary}{Corollary}[section]
\newtheorem{lemma}{Lemma}[section]
\newtheorem{claim}{Claim}[section]
\newtheorem{definition}{Definition}[section]

\usepackage{framed}

\newenvironment{proofof}[1]{\smallskip
\noindent {\bf Proof of #1.  }}{\hfill$\Box$
\smallskip}

\newenvironment{reminder}[1]{\smallskip
\noindent {\bf Reminder of #1  }\em}{
\smallskip}

\newenvironment{proof}{\noindent {\bf Proof.  }}{\hfill$\Box$
\smallskip}

\def \iouseful {\text{useful}}

\def \coRP {{\sf coRP}}

\def \SUBEXP {{\sf SUBEXP}}
\def \ZPSUBEXP {{\sf ZPSUBEXP}}

\def \RP {{\sf RP}}

\def \REXP {{\sf REXP}}
\def \coNP {{\sf coNP}}
\def \BPP {{\sf BPP}}
\def \NC {{\sf NC}}
\def \ZPE {{\sf ZPE}}
\def \NE {{\sf NE}}
\def \E {{\sf E}}
\def \poly { \text{\rm poly} }
\def \TC {{\sf TC}}

\def \Z {{\mathbb Z}}
\def \P {{\sf P}}
\def \MA {{\sf MA}}
\def \AM {{\sf AM}}
\def \MATIME {{\sf MATIME}}

\def \NP {{\sf NP}}
\def \EXP {{\sf EXP}}

\def \T {{\sf TIME}}
\def \TIME {\T}
\def \Sig[#1] {{\sf \Sigma}_{#1} }
\def \Pie[#1] {{\sf \Pi}_{#1} }

\def \ACC {{\sf ACC}^0}

\def \io {\textrm{\it io-}}

\def \SIZE {{\sf SIZE}}

\def \AC {{\sf AC}}
\def \BPTIME {{\sf BPTIME}}

\def \RE {{\sf RE}}
\def \ZPSUBEXP {{\sf ZPSUBEXP}}

\def \RTIME {{\sf RTIME}}
\def \NSUBEXP {{\sf NSUBEXP}}

\def \NT {{\sf NTIME}}

\def \NTIME {\NT}

\def \coNT {{\sf coNTIME}}

\def \coNTIME {\coNT}

\def \ZPTIME {{\sf ZPTIME}}
\def \ZPP  {{\sf ZPP}}

\def \coNE {{\sf coNE}}

\def \N {{\mathbb N}}

\def \NEXP {{\sf NEXP}}

\def \SIZE {{\sf SIZE}}

\def \eps {\varepsilon}

\title{Natural Proofs Versus Derandomization\footnote{A preliminary version of this paper appeared in the ACM Symposium on Theory of Computing in 2013.}}
\author{
Ryan Williams\thanks{Computer Science Department, Stanford University, {\tt rrw@cs.stanford.edu}. Supported in part by a David Morgenthaler II Faculty Fellowship, a Sloan Fellowship, NSF DMS-1049268 (US Junior Oberwolfach Fellow), and NSF CCF-1212372. Any opinions, findings and conclusions or recommendations expressed in this material are those of the authors and do not necessarily reflect the views of the National Science Foundation.}
}

\date{}

\begin{document}

\maketitle

\begin{abstract} We study connections between Natural Proofs, derandomization, and the problem of proving ``weak'' circuit lower bounds such as $\NEXP \not\subset \TC^0$, which are still wide open.

Natural Proofs have three properties: they are \emph{constructive} (an efficient algorithm $A$ is embedded in them), have \emph{largeness} ($A$ accepts a large fraction of strings), and are \emph{useful} ($A$ rejects all strings which are truth tables of small circuits). Strong circuit lower bounds that are ``naturalizing'' would contradict present cryptographic understanding, yet the vast majority of known circuit lower bound proofs are naturalizing. So it is imperative to understand how to pursue un-Natural Proofs. Some heuristic arguments say  constructivity should be circumventable: largeness is inherent in many proof techniques, and it is probably our presently weak techniques that yield constructivity. We prove:

\smallskip

\noindent $\bullet$ \emph{Constructivity is unavoidable}, even for $\NEXP$ lower bounds. Informally, we prove for all ``typical'' non-uniform circuit classes ${\cal C}$, $\NEXP \not\subset {\cal C}$ if and only if there is a polynomial-time algorithm distinguishing \emph{some} function from all functions computable by ${\cal C}$-circuits. Hence $\NEXP \not\subset {\cal C}$ is equivalent to exhibiting a constructive property useful against ${\cal C}$.
\smallskip

\noindent $\bullet$ There are no $\P$-natural properties useful against ${\cal C}$ if and only if randomized exponential time can be ``derandomized'' using truth tables of circuits from ${\cal C}$ as random seeds. Therefore the task of proving there are no $\P$-natural properties is inherently a \emph{derandomization} problem, weaker than but implied by the existence of strong pseudorandom functions.

These characterizations are applied to yield several new results, including improved $\ACC$ lower bounds and new unconditional derandomizations. In general, we develop and apply several new connections between the existence of certain algorithms for analyzing truth tables, and the \emph{non-existence} of small circuits for problems in large classes such as $\NEXP$.
\end{abstract}

\section{Introduction}

The Natural Proofs barrier of Razborov and Rudich~\cite{RazborovRudich97} argues that\begin{itemize}
\item[(a)] almost all known proofs of non-uniform circuit lower bounds entail efficient algorithms that can distinguish many ``hard'' functions from all ``easy'' functions (those computable with small circuits), and
\item[(b)] any efficient algorithm of this kind would break cryptographic primitives implemented with small circuits (which are believed to exist).
\end{itemize}
(A formal definition is in Section~\ref{prelims}.) Natural Proofs are self-defeating: in the course of proving a weak lower bound, they provide efficient algorithms that refute stronger lower bounds that we believe to also hold. The moral is that, in order to prove stronger circuit lower bounds, one must avoid the techniques used in proofs that entail such efficient algorithms. The argument applies even to low-level complexity classes such as $\TC^0$~\cite{Naor-Reingold04,Krause-Lucks01,Miles-Viola}, so any major progress in the future depends on proving un-Natural lower bounds.

How should we proceed? Should we look for proofs yielding only \emph{inefficient} algorithms, avoiding ``constructivity''? Or should we look for algorithms which cannot distinguish \emph{many} hard functions from all easy ones, avoiding ``largeness''?\footnote{See the webpage~\cite{Scottblog} for a discussion with many views on these questions.} (Note there is a third criterion, ``usefulness'', requiring that the proof distinguishes a target function $f$ from the circuit class ${\cal C}$ we are proving lower bounds against. This criterion is necessary: $f \notin {\cal C}$ if and only if there is a trivial property, true of only $f$, distinguishing $f$ from all functions computable in ${\cal C}$.) 
In this paper, 
we study alternative ways to characterize Natural Proofs and their relatives as particular circuit lower bound problems, and give several  applications. There are multiple competing intuitions about the meaning of Natural Proofs. We wish to rigorously understand the extent to which the Razborov-Rudich framework relates to our ability to prove lower bounds in general.

\paragraph{NEXP lower bounds are constructive and useful} Some relationships can be easily seen. Recall $\EXP$ and $\NEXP$ are the exponential-time versions of $\P$ and $\NP$. If $\EXP \not\subset {\cal C}$, one can obtain a polynomial-time (non-large) property useful against ${\cal C}$.\footnote{Define $A(T)$ to accept its $2^n$-bit input $T$ if and only if $T$ is the truth table of a function that is complete for $\E = \TIME[2^{O(n)}]$. $A$ can be implemented to run in $\poly(2^n)$ time and rejects all $T$ with ${\cal C}$ circuits, assuming $\EXP \not\subset {\cal C}$.} So, strong enough lower bounds entail constructive useful properties. However, a separation like $\EXP \not\subset {\cal C}$ is stronger than currently known, for all classes ${\cal C}$ containing $\ACC$. Could lower bounds be proved for larger classes like $\NEXP$, without entering constructive/useful territory? In the other direction, could one exhibit a constructive (non-large) property against a small circuit class like $\TC^0$, without proving a new lower bound against that class?

The answer to both questions is \emph{no}. Call a (non-uniform) circuit class ${\cal C}$ \emph{typical} if ${\cal C} \in \{\AC^0$, $\ACC$, $\TC^0$, $\NC^1$, $\NC$, $\P/\poly\}$.\footnote{For simplicity, in this paper we mostly restrict ourselves to typical classes; however it will be clear from the proofs that we only rely on a few properties of these classes, and more general statements can be made.} For any typical ${\cal C}$, a property of Boolean functions ${\cal P}$ is said to be $\emph{\iouseful}$ against ${\cal C}$ if, for all $k$, there are infinitely many $n$ such that

\smallskip

\begin{compactitem}
\item ${\cal P}(f)$ is true of at least one $f : \{0,1\}^n\rightarrow \{0,1\}$, and
\item  ${\cal P}(g)$ is false for all $g:\{0,1\}^n \rightarrow \{0,1\}$ having $n^k$ size ${\cal C}$-circuits.
\end{compactitem}

\smallskip

In other words, on infinitely many input lengths $n$, ${\cal P}$ distinguishes some function from all easy functions. We prove: 

\begin{theorem}%[Section~\ref{nexp}]
\label{equiv} For all typical ${\cal C}$, $\NEXP \not\subset {\cal C}$ if and only if there is a polynomial-time computable property of Boolean functions that is $\iouseful$ against ${\cal C}$ with $O(\log n)$ bits of advice.
\end{theorem}

That is, $\NEXP \not\subset {\cal C}$ if and only if there is a language in $\P/O(\log n)$ defining a property of Boolean functions useful against ${\cal C}$. 

We can remove the $O(\log n)$ bits of advice of Theorem~\ref{equiv} by relaxing the notion of a ``property'' of Boolean functions to hold over all strings. Boolean function properties are only defined on $2^n$-length binary strings; however, \emph{every} binary string $x$ can be viewed as the truth table of a unique Boolean function, by simply appending zeroes to the end of $x$ until its length is a power of $2$. For brevity we shall call this longer string $f_x$, which is a function from $\{0,1\}^{\ell}$ to $\{0,1\}$ where $\ell$ is the smallest integer satisfying $2^{\ell} \geq |x|$. Informally, we define an \emph{algorithm} $A$ to be $\emph{\iouseful}$ against ${\cal C}$ if, for all $k$, there are infinitely many input lengths $N$ such that

\smallskip

\begin{compactitem}
\item for at least one $x \in \{0,1\}^N$, $A(x) = 1$, and 
\item for all $x' \in \{0,1\}^{N}$ such that $f_{x'}:\{0,1\}^{\ell} \rightarrow \{0,1\}$ has $n^k$ size ${\cal C}$-circuits, $A(x') = 0$.
\end{compactitem}

\smallskip

\begin{theorem}%[Section~\ref{nexp}]
\label{equiv2} For all typical ${\cal C}$, $\NEXP \not\subset {\cal C}$ if and only if there is a polynomial-time \emph{algorithm} that is $\iouseful$ against ${\cal C}$.
\end{theorem}

Theorems~\ref{equiv} and \ref{equiv2} help explain \emph{why} it is difficult to prove even $\NEXP$ circuit lower bounds: \emph{any} $\NEXP$ lower bound must meet precisely two of the three conditions of Natural Proofs (constructivity and usefulness).\footnote{One may also wonder if \emph{non-constructive large properties} imply any new circuit lower bounds. This question does not seem to be as interesting. For one, there are already $\coNP$-natural properties useful against $\P/\poly$ (simply try all possible small circuits in parallel), and the consequences of such properties are well-known. So anything $\coNP$-constructive or worse is basically uninformative (without further information on the property). Furthermore, slightly more constructive properties, such as $\NP$-natural ones, seem unlikely~\cite{Rudich97}.} The above two theorems say that \emph{every} $\NEXP$ circuit lower bound must exhibit some constructive property useful against those circuits. Polynomial-time algorithms distinguishing ``some'' functions from ``all'' easy functions look difficult to construct, even infinitely often; if one adds in largeness too, these algorithms are likely \emph{impossible} to construct.

One can make a heuristic argument that the recent proof of $\NEXP \not\subset \ACC$ (\cite{Williams11}) evades Natural Proofs by being non-constructive. Intuitively, the proof uses an $\ACC$ Circuit SAT algorithm that only mildly improves over brute force, so it runs too slowly to obtain a polytime property useful against $\ACC$. Theorem~\ref{equiv} shows that, in fact, constructivity is necessary. 
Moreover, the proofs of Theorems~\ref{equiv} and \ref{equiv2} yields an explicit property $\iouseful$ against $\ACC$. 

The techniques used in these theorems can be applied, along with several other ideas, to prove new super-polynomial lower bounds against $\ACC$. First, we prove exponential-size lower bounds on the $\ACC$ circuit complexity of encoding \emph{witnesses} for $\NEXP$ languages. 

\begin{theorem}\label{expwitnessACC} For all $d$, $m$ there is an $\eps>0$ such that $\NTIME[2^{O(n)}]$ does not have $2^{n^{\eps}}$-size $d$-depth $\AC^0[m]$ witnesses.
\end{theorem}

Formal definitions can be found in Section~\ref{nexp}; informally, Theorem~\ref{expwitnessACC} says that there are $\NEXP$ languages with verifiers that only accept witness strings of exponentially high $\ACC$ circuit complexity. It is interesting that while we can prove such lower bounds for encoding $\NEXP$ witnesses, we do not yet know how to prove them for $\NEXP$ languages themselves (the best known size lower bound for $\NEXP$ is ``third-exponential''). 

These circuit lower bounds for witnesses can also be translated into new $\ACC$ lower bounds for some complexity classes. Recall that $\NE = \NTIME[2^{O(n)}]$ and $\io\coNE = \io\coNTIME[2^{O(n)}]$, the latter being the class of languages $L$ such that there is an $L' \in \coNTIME[2^{O(n)}]$ where, for infinitely many $n$, $L \cap \{0,1\}^n = L' \cap \{0,1\}^n$. That is, $L$ agrees with a language in $\coNTIME[2^{O(n)}]$ on infinitely many input lengths. The class $\NE/1 \cap \coNE/1$ consists of languages $L \in \NE \cap \coNE$ recognizable with ``one bit of advice.'' That is, there are nondeterministic machines $M$ and $M'$ running in $2^{O(n)}$ time with the property that for all $n$, there are bits $y_n,z_n \in \{0,1\}$ such that for all strings $x$, $x \in L$ if and only if $M(x,y_n)$ accepts on all paths if and only if $M'(x,z_n)$ rejects on all paths. (In fact, in our case we may assume $y_n=z_n$ for all $n$.)

\begin{theorem}\label{NEcapcoNE}
$\NE \cap \io\coNE$ and $\NE/1 \cap \coNE/1$ do not have $\ACC$ circuits of $n^{\log n}$ size.\footnote{This is not the strongest size lower bound that can be proved, but it is among the cleanest. Please note that the conference version of this paper claimed a lower bound for the (hypothetically smaller) class $\NE \cap \coNE$; we are grateful to Russell Impagliazzo and Igor Carboni Oliveira~\cite{Oliveira13} for observing that our argument only proves a lower bound for $\NE \cap \io\coNE$ (and $\NE \cap \coNE$ with one bit of advice, under the appropriate definition).}
\end{theorem}

These lower bound are intriguing because they necessarily must be proved differently. The known proof of $\NEXP \not\subset \ACC$ works for the class $\NEXP$ because there is a tight time hierarchy for nondeterminism~\cite{Zak83}. However, the $\NTIME \cap \coNTIME$ classes (and $\NTIME \cap \io\coNTIME$ classes) are not known to have such a hierarchy. (They are among the ``semantic'' classes, which are generally not known to have complete languages or nice time hierarchies.) Interestingly, the proof of Theorem~\ref{NEcapcoNE} crucially uses the previous lower bound framework against $\NEXP$, and builds on it, via Theorem~\ref{equiv} and a modification of the $\NEXP \not\subset \ACC$ lower bound. Indeed, it follows from the arguments here (building on~\cite{Williams10,Williams11}) that the lower bound consequences of non-trivial circuit SAT algorithms can be strengthened, in the following sense:

\begin{theorem}\label{generic}  Let ${\cal C}$ be typical. Suppose the satisfiability problem for $n^{O(\log^c n)}$-size ${\cal C}$ circuits can be solved in $O(2^n/n^{10})$ time, for all constants $c$. Then $\NE \cap \io\coNE$ and $\NE/1 \cap \coNE/1$ do not have $n^{\log n}$-size ${\cal C}$ circuits.
\end{theorem}

\begin{theorem}\label{generic2} Suppose we can approximate the acceptance probability of any given $n^{O(\log^c n)}$-size circuit (with fan-in two and arbitrary depth) on $n$ inputs to within $1/6$, for all $c$, in $O(2^n/n^{10})$ time (even nondeterministically). Then $\NE \cap \io\coNE$ and $\NE/1 \cap \coNE/1$ do not have $n^{\log n}$-size circuits.
\end{theorem}

\paragraph{Natural Proofs vs Derandomization}
Given Theorem~\ref{equiv}, it is natural to wonder if full-strength natural properties are equivalent to some circuit lower bound problems. If so, such lower bounds should be considered \emph{unlikely}. To set up the discussion, let $\RE = \RTIME[2^{O(n)}]$ and $\ZPE = \ZPTIME[2^{O(n)}]$; that is, $\RE$ is the class of languages solvable in $2^{O(n)}$ randomized time with one-sided error, and $\ZPE$ is the corresponding class with zero error (i.e., \emph{expected} $2^{O(n)}$ running time).

For a typical circuit class ${\cal C}$, we informally say that $\RE$ (respectively, $\ZPE$) \emph{has ${\cal C}$ seeds} if, for every predicate defining a language in the respective complexity class, there are ${\cal C}$ circuit families succinctly encoding exponential-length ``seeds'' that correctly decide the predicate. (Formal definitions are given in Section~\ref{natural}.) Having ${\cal C}$ seeds means that the randomized class can be derandomized very strongly: by trying all poly-size ${\cal C}$ circuits as random seeds, one can decide any predicate from the class in $\EXP$.

We prove a strong correspondence between the existence of such seeds, and the \emph{nonexistence} of natural properties:

\begin{theorem}\label{naturalequiv} Let ${\cal C}$ be typical. The following are equivalent:
\begin{compactenum}
\item There are no $\P$-natural properties $\iouseful$ (respectively, ae-useful\footnote{Here, \emph{ae-useful} is just the ``almost-everywhere useful'' version, where the property is required to distinguish random functions from easy ones on almost every input length.}) against ${\cal C}$
\item $\ZPE$ has ${\cal C}$ seeds for almost all (resp., infinitely many) input lengths
\end{compactenum}
\end{theorem}

One can remove the $O(\log n)$ advice similarly to Theorem~\ref{equiv2} by relaxing the ``property of Boolean functions'' to algorithms on arbitrary strings. 
Informally, Theorem~\ref{naturalequiv} says that ruling out $\P$-natural properties is equivalent to a strong derandomization of randomized exponential time, using small circuits to encode exponentially-long random seeds. 
Similarly, we prove that a variant of natural properties is related to succinct ``hitting sets'' for $\RE$~(Theorem~\ref{rtimenatural}). 

It is worth discussing the meaning of these results in a little more detail. Let ${\cal C}, {\cal D}$ be appropriate circuit classes. 
Roughly speaking, the key lesson of Natural Proofs~\cite{RazborovRudich97,Naor-Reingold04,Krause-Lucks01} is that, if there are ${\cal D}$-natural properties useful against ${\cal C}$, then there are no pseudorandom functions (PRFs) computable in ${\cal C}$ that fool ${\cal D}$ circuits; namely, 
there is a statistical test $T$ computable in ${\cal D}$ such that, for every function $f(\cdot,\cdot) \in {\cal C}$, the test $T$ with query access to $f(x,\cdot)$ (where $x$ is a uniform random $n$-bit seed) can distinguish $f(x,\cdot)$ from a uniform random function (generated using $2^n$ uniform random bits). Now, if we have a PRF computable in ${\cal C}$ that can fool ${\cal D}$ circuits, this PRF can be used to obtain ${\cal C}$ seeds for randomized ${\cal D}$ circuits with one-sided error.\footnote{Consider any ${\cal D}$-circuit $D$ that tries to use $f$ as a source of randomness. A ${\cal C}$-circuit seed for $D$ can be obtained from a circuit computing $f$: since $f$ fools $D$, at least one $n$-bit seed to $f$ will make $D^f$ print $1$.}  That is, the existence of PRFs implies the existence of ${\cal C}$ seeds, so our consequence in Theorem~\ref{naturalequiv} (of the existence of natural properties) that ``no $\ZPE$ predicate has ${\cal C}$ seeds'' appears stronger than ``there are no PRFs'' (as in~\cite{RazborovRudich97}). Moreover, this stronger consequence in Theorem~\ref{naturalequiv} (and Theorem~\ref{rtimenatural}, proved later) yields an implication in the reverse direction: the \emph{lack} of ${\cal D}$-natural properties implies strong derandomizations of randomized 
exponential-size ${\cal D}$. 

Theorem~\ref{naturalequiv} also shows that plausible some derandomization problems are as hard as resolving $\P \neq \NP$.  Since we believe that there are no $\P$-natural properties useful against $\P/\poly$, then by Theorem~\ref{naturalequiv}, we must also believe that there are  ``canonical'' derandomizations of $\ZPE$ in $\EXP$, along the lines of item (2) in Theorem~\ref{naturalequiv}. However, \emph{proving} that such a canonical derandomization exists would in turn imply that there are no $\P$-natural properties useful against $\P/\poly$ (again by Theorem~\ref{naturalequiv}) and hence $\P \neq \NP$. 

\paragraph{Unconditional mild derandomizations} Understanding the relationships between the randomized complexity classes $\ZPP$, $\RP$, and $\BPP$ is a central problem in modern complexity theory. It is well-known that \[\P \subseteq \ZPP = \RP \cap \coRP \subseteq \RP \subseteq \BPP\] but it is not known if any inclusion is an equality. The ideas behind Theorem~\ref{naturalequiv} can also be applied to prove new relations between these classes. We define ${\sf ZPTIME}[t(n)]/d(n)$ to be the class of languages solvable in zero-error time $t(n)$ by machines of description length at most $d(n)$ (under some standard encoding of machines).\footnote{{\bf N.B.} Although our definition is standard (see for example \cite{Barak02,FST05}), it is important to note that there are other possible interpretations of the same notation. Here, we only require that the algorithm is required to be zero-error for the ``correct'' advice or description, but one could also require that the algorithm is zero-error \emph{no matter what} advice is given.} The ``infinitely often'' version $\io{\sf ZPTIME}[t(n)]/d(n)$ is the class of languages $L$ solvable with machines of description length $d(n)$ running in time $t(n)$ that are zero-error for infinitely many input lengths: for infinitely many $n$, the machine has the zero-error property on all inputs of length $n$.
 
\begin{theorem}\label{derand} Either $\RTIME[2^{O(n)}] \subseteq \SIZE[n^c]$ for some $c$, or $\BPP \subseteq \io\ZPTIME[2^{n^{\eps}}]/n^{\eps}$ for all $\eps > 0$.
\end{theorem}

We have a win-win: either randomized exptime is very easy with non-uniform circuits, or randomized computation with two-sided error has a  \emph{zero error} simulation (with description size $n^{\eps}$) that dramatically avoids brute-force. To appreciate the theorem statement, suppose the first case could be modified to conclude that $\RP \subseteq  \io\ZPTIME[2^{n^{\eps}}]/n^{\eps}$ for all $\eps > 0$. Then the famous ($\coRP$) problem of Polynomial Identity Testing would have a new subexponential-time algorithm, good enough to prove strong $\NEXP$ circuit lower bounds.\footnote{More precisely, the main result of Kabanets and Impagliazzo~\cite{KI04} concerning the derandomization of Polynomial Identity Testing (PIT) can be extended as follows: if PIT for arithmetic circuits can be solved for infinitely-many circuit sizes in nondeterministic subexponential time, then either $\NEXP \not\subset \P/\poly$ or the Permanent does not have polynomial-size arithmetic circuits.} A quick corollary of Theorem~\ref{derand} comes close to achieving this. To simplify notation, we use the $\SUBEXP$ modifier in a complexity class to abbreviate ``$2^{n^{\eps}}$ time, for every $\eps > 0$.''

\begin{corollary} \label{rp} For some $c > 0$, $\RP \subset \io\ZPSUBEXP/n^c$.\end{corollary}

That is, the error in an $\RP$ computation can be removed in subexponential time with fixed-polynomial advice, infinitely often. We emphasize that the advice needed is \emph{independent} of the running times of the $\RP$ and $\ZPSUBEXP$ computations: the $\RP$ computation could run in $n^{c^{{c}^{{c}^c}}}$ time and still need only $n^c$ advice to be simulated in $2^{n^{1/{c^{{c}^{{c}^c}}}}}$ time. 
Corollary~\ref{rp} should be compared with a theorem of Kabanets~\cite{Kabanets01}, who gave a simulation of $\RP$ in \emph{pseudo}-subexponential time with zero error. That is, his simulation is only guaranteed to succeed against efficient adversaries which try to generate bad inputs (but his simulation also does not require advice).

An analogous argument can be used to give a new simulation of Arthur-Merlin games. Informally (and following the notation outlined above), $\io\Sigma_2 \SUBEXP/n^c$ is the class of languages which agree infinitely often with $\Sigma_2$ machines running in $2^{n^{\eps}}$ time for all $\eps > 0$, with $O(n^c)$ bits of advice.

\begin{corollary} \label{am} For some $c > 0$, $\AM \subseteq \io\Sigma_2 \SUBEXP/n^c$.
\end{corollary}

The ideas used here can also be applied to prove a new equivalence between $\NEXP = \BPP$ and nontrivial simulations of $\BPP$. Informally, $\io{\sf Heuristic}\ZPTIME[2^{n^{\eps}}]/n^{\eps}$ is the class of languages which, for infinitely many $n$, agree on a $1-1/n$ fraction of the $n$-bit inputs with zero-error randomized subexponential-time machines using $O(n^{\eps})$ advice.

\begin{theorem}\label{NEXPBPPequiv}
$\NEXP \neq \BPP$ if and only if for all $\eps > 0$, $\BPP \subseteq \io{\sf Heuristic}\ZPTIME[2^{n^{\eps}}]/n^{\eps}$.
\end{theorem}

Finally, these ideas can be extended to show an \emph{equivalence} between the existence of $\RP$-natural properties and $\P$-natural properties against a circuit class:

\begin{theorem}\label{RPnaturalprop} If there exists a $\RP$-natural property $P$ useful against a class ${\cal C}$, then there exists a  $\P$-natural property $P'$ against ${\cal C}$.
\end{theorem}

That is, given any property $P$ with one-sided error that is sufficient for distinguishing all easy functions from many hard functions, we can obtain a  deterministic property $P'$ with analogous behavior. (Note this is not \emph{exactly} a derandomization of property $P$; the property $P'$ will in general have different input-output behavior from $P$, but $P'$ does use $P$ as a subroutine.) The key idea of the proof is to \emph{swap the input with the randomness} in the property $P$.

\section{Preliminaries}\label{prelims}

For simplicity, all languages are over $\{0,1\}$. We fix some standard encoding of Turing machines, and define the \emph{description length} of a machine $M$ to be the length of $M$ under the encoding. We assume knowledge of the basics of complexity theory~\cite{AroraBarak} such as advice-taking machines, and complexity classes like $\EXP=\TIME[2^{n^{O(1)}}]$, $\NEXP = \NTIME[2^{n^{O(1)}}]$, $\AC^0[m]$, $\ACC$, and so on. We use $\SIZE[s(n)]$ to denote the class of languages recognized by a (non-uniform) $s(n)$-size circuit family. We also use the (standard) ``subexponential-time'' notation $\SUBEXP = \bigcap_{\eps > 0}\TIME[2^{O(n^{\eps})}]$. (So for example, $\NSUBEXP$ refers to the class of languages accepted in nondeterministic $2^{n^{\eps}}$ time, for all $\eps > 0$.) When we refer to a ``typical'' circuit class ($\AC^0$, $\ACC$, $\TC^0$, $\NC^1$, $\NC$, or $\P/\poly\}$), we will always assume the class is \emph{non-uniform}, unless otherwise specified. Some familiarity with prior work connecting SAT algorithms and circuit lower bounds~\cite{Williams10,Williams11} would be helpful, but this paper is mostly self-contained.

We will use \emph{advice classes}: for a deterministic or nondeterministic class ${\cal C}$ and a function $a(n)$, ${\cal C}/a(n)$ is the class of languages $L$ such that there is an $L' \in {\cal C}$ and an arbitrary function $f : \N \rightarrow \{0,1\}^{\star}$ with $|f(n)|\leq a(n)$ for all $x$, such that $L = \{x~|~(x,f(|x|)) \in L'\}$. That is, the arbitrary advice string $f(n)$ can be used to solve all $n$-bit instances within class  ${\cal C}$. 

For semantic (e.g., randomized, $\NTIME \cap \coNTIME$) classes ${\cal C}$, the definition of advice is technically subtle. We shall only require that the class ${\cal C}$ algorithm exhibits the relevant promise condition (zero-error, one-sided error, or otherwise) for the ``correct'' advice or description; one could also require that the algorithm satisfies the promise condition \emph{no matter what} advice is given. 

More precisely, for a randomized machine $M$ and class ${\cal C} \in \{\RTIME[t(n)],\ZPTIME[t(n)],\BPTIME[t(n)]\}$, we say that $M$ is of type ${\cal C}$ on a given input $x$ if $M$ on $x$ runs in time $t(|x|)$ and $M$ satisfies the promise of one-sided/zero/two-sided error on input $x$. (For example, in the case of one-sided error, if $x \in L$ then $M$ on $x$ should accept at least $2/3$ of the computation paths; if $x \notin L$ then $M$ on $x$ should reject all of the computation paths. In the case of zero-error, if $x \in L$ then $M$ on $x$ should accept at least $2/3$ of the paths and output {\bf ?} (i.e., \emph{don't know}) on the others; if $x \notin L$ then $M$ on $x$ should reject at least $2/3$ of the paths and output {\bf ?} on the others.) Then for ${\cal C} \in \{\RTIME[t(n)],\ZPTIME[t(n)],\BPTIME[t(n)]\}$, ${\cal C}/a(n)$ is the class of languages $L$ recognized by a randomized machine of \emph{description length} $a(n)$ (under some standard encoding of machines) that is of type ${\cal C}$ on all inputs~\cite{Barak02}. Equivalently, $L \in {\cal C}/a(n)$ is in the class if there is a machine $M$ and advice function $s : \N \rightarrow \{0,1\}^{a(n)}$ such that for all $x \in \{0,1\}^{\star}$, $M$ is a machine of type ${\cal C}$ when executed on input $(x,a(|x|)$ ($M$ satisfies the promise of one-sided/zero/two-sided error on that input) and $x \in L$ if and only if $M(x,a(|x|))$ accepts~\cite{FST05}. 

We also use \emph{infinitely-often classes}: for a deterministic or nondeterministic complexity class ${\cal C}$, $\io{\cal C}$ is the class of languages $L$ such that there is an $L' \in {\cal C}$ where, for infinitely many $n$, $L \cap \{0,1\}^n = L' \cap \{0,1\}^n$. For randomized classes ${\cal C} \in \{\RTIME[t(n)],\ZPTIME[t(n)],\BPTIME[t(n)]\}$, as well as semantic classes such as $(\NTIME\cap\coNTIME)[t(n)]$, $\io{\cal C}$ is the class of languages $L$ recognized by a machine $M$ such that, for infinitely many input lengths $n$, $M$ is of type ${\cal C}$ on all inputs of length $n$ (and need not be of type ${\cal C}$ on other input lengths). 

Some particular notation and conventions will be useful for this paper. For any circuit $C(x_1,x_2,\ldots,x_n)$, $i < j$, and $a_1,\ldots,a_n \in \{0,1\}$, the notation $C(a_1,\ldots,a_i,\cdot,a_j,\ldots,a_n)$ represents the circuit with $j-i-1$ inputs obtained by assigning the input $x_q$ to $a_q$, for all $q \in [1,i]\cup[j,n]$. In general, $\cdot$ is used to denote free unassigned inputs to the circuit.

\subsection{Truth Tables and Their Circuit Complexity}\label{tt-ckt}

In this paper, we study the circuit complexities of all strings, even those which are not of length equal to a power of two. To make the discussion precise, we carefully develop the concepts in this section.

Let $y_1,\ldots,y_{2^k} \in \{0,1\}^k$ be the list of $k$-bit strings in lex order. For a Boolean function $f:\{0,1\}^n \rightarrow \{0,1\}$, the truth table of $f$ is defined to be \[tt(f) := f(y_1)f(y_2) \cdots f(y_{2^n}),\] and the truth table of a circuit is simply the truth table of the function it defines. For binary strings with lengths that are not powers of two, we use the following encoding convention. Let $T$ be a binary string, let $k=\lceil \log_2 |T| \rceil$. The \emph{Boolean function encoded by $T$} or the \emph{function corresponding to $T$}, denoted by $f_T$, is the function satisfying $tt(f_T) = T0^{2^k-|T|}$. 

The \emph{size} of a circuit is its number of gates. The circuit complexity of an arbitrary string (and hence, a function) takes some care to properly define, based on the circuit model. For the unrestricted model, the \emph{circuit complexity of $T$}, denoted as $CC(T)$, is simply the minimum size of any circuit computing $f_T$. For a depth-bounded circuit model, where a depth function must be specified prior to giving the circuit family, the appropriate measure is the \emph{depth-$d$ circuit complexity of $T$}, denoted as $CC_d(T)$, which is the minimum size of any depth-$d$ circuit computing $f_T$. (Note that, even for circuit classes like $\NC^1$, we have to specify a depth upper bound $c \log n$ for some constant $c$.) For the class $\ACC$, we must specify a modulus $m$ for the MOD gates, as well as a depth bound, so when considering $\ACC$ circuit complexity, we look at the \emph{depth-$d$ mod-$m$ circuit complexity of $T$}, $CC_{d,m}(T)$, for fixed $d$ and $m$.

A simple fact about the circuit complexities of truth tables and their substrings will be very useful:

\begin{proposition}\label{truthtable0} Suppose $T = T_1\cdots T_{2^k}$ is a string of length $2^{k+\ell}$, where $T_1,\ldots,T_{2^k}$ each have length $2^{\ell}$. Then $CC(T_i) \leq CC(T)$, $CC_{d}(T_i) \leq CC_d(T)$, and $CC_{d,m}(T_i) \leq CC_{d,m}(T)$.
\end{proposition}

\begin{proof} Given a circuit $C$ of size $s$ for $f_T$, a circuit for $f_{T_i}$ is obtained by substituting values for the first $k$ inputs of $C$. This yields a circuit of size at most $s$.
\end{proof}

We will sometimes need a more general claim: for any string $T$, the circuit complexity of an arbitrary substring of $T$ can be bounded via the circuit complexity of $T$.

\begin{lemma}\label{truthtable} There is a universal $c \geq 1$ such that the following holds. Let $T$ be a binary string, and let $S$ be any substring of $T$. Then for all $d$ and $m$, $CC(f_{S}) \leq CC(f_T)+(c\log|T|)$, $CC_{d}(f_{S}) \leq  CC_{d+c}(f_T)+(c\log|T|)^{1+o(1)}$, and $CC_{d,m}(f_{S}) \leq CC_{d+c,m}(f_T)+(c\log|T|)^{1+o(1)}$.
\end{lemma}

\begin{proof} Let $c'$ be sufficiently large in the following. Let $k$ be the minimum integer satisfying $2^k \geq |T|$, so the Boolean function $f_T$ representing $T$ has truth table $T0^{2^k-|T|}$. Suppose $C$ is a size-$s$ depth-$d$ circuit for $f_T$. Let $S$ be a substring of $T=t_1\cdots t_{2^k} \in \{0,1\}^{2^k}$, and let $A, B \in \{1,\ldots,2^k\}$ be such that $S = t_A \cdots t_B$. Let $\ell \leq k$ be a minimum integer which satisfies $2^{\ell} \geq B-A$. Our goal is to construct a small circuit $D$ with $\ell$ inputs and truth table $S 0^{2^{\ell}-(B-A)}$.

Let $x_1,\ldots,x_{2^{\ell}}$ be the $\ell$-bit strings in lex order. The desired circuit $D$ on input $x_i$ can be implemented as follows: Compute $i+A$. If $(i+A) \leq B$ then output $C(x_{i+A})$, otherwise output $0$. To bound the size of $D$, first note there are depth-$c'$ circuits of at most $c'\cdot n \log^{\star} n$ size for addition of two $n$-bit numbers~\cite{Chandra-Fortune-Lipton85}, and there are also well-known $O(n)$-size (unrestricted depth) circuits for addition. 

Therefore in depth-$c'$ and size at most $c' \cdot k \log^{\star} k$ we can, given input $x_i$ of length $\ell$, output $i+A$. Determining if $i \leq B-A$ can be done with $(c' \cdot \ell)$-size depth-$c'$ circuits. Therefore $D$ can either be implemented as a circuit of size at most $s + c'((k\log^{\star} k) +\ell+1)$ and depth $2c'+d$, or as an (unrestricted depth) circuit of size at most $s + c'(k+\ell+1)$. To complete the proof, let $c \geq 3c'$.
\end{proof}

We will use the following strong construction of pseudorandom generators from hard functions:

\begin{theorem}[Umans~\cite{Umans03}]\label{hardness-randomness} There is a universal constant $g$ and a function $G : \{0,1\}^{\star} \times \{0,1\}^{\star} \rightarrow \{0,1\}^{\star}$ such that, for all $s$ and $Y$ satisfying $CC(Y) \geq s^g$, and for all circuits $C$ of size $s$, \[\left|\Pr_{x \in \{0,1\}^{g \log |Y|}}[C(G(Y,x))=1] - \Pr_{x \in \{0,1\}^s}[C(x)=1]\right| < 1/s.\] Furthermore, $G$ is computable in $\poly(|Y|)$ time.
\end{theorem}

\paragraph{Natural Proofs} A \emph{property of Boolean functions ${\cal P}$} is a subset of the set of all Boolean functions. Let $\Gamma$ be a complexity class and let ${\cal C}$ be a circuit class (typically, $\Gamma=\P$ and ${\cal C}=\P/\poly$). A \emph{$\Gamma$-natural property useful against ${\cal C}$} is a property of Boolean functions ${\cal P}$ that satisfies the axioms:
\begin{compactitem}
\item (Constructivity) \indent~~~\indent ${\cal P}$ is decidable in $\Gamma$, 
\item (Largeness)  \indent~~~\indent for all $n$, ${\cal P}$ contains a $1/2^{O(n)}$ fraction of all $2^n$-bit strings,
\item (Usefulness)  \indent~~~\indent Let $f = \{f_n\}$ be a sequence of functions $\{f_n\}$ such that $f_n \in {\cal P}$ for all $n$. Then for all $k$ and infinitely many $n$, $f_n$ does not have $n^k$-size ${\cal C}$-circuits.\footnote{Note that some papers, including Razborov and Rudich~\cite{RazborovRudich97}, replace `infinitely many' with `almost every'; in this paper, we call that version \emph{ae-usefulness}.}
\end{compactitem}

Let $f = \{f_n : \{0,1\}^n \rightarrow \{0,1\}\}$ be a sequence of Boolean functions. A \emph{$\Gamma$-natural proof} that $f \not\in {\cal C}$ establishes the existence of a $\Gamma$-natural property ${\cal P}$ useful against ${\cal C}$ such that ${\cal P}(f_n)=1$ for all $n$. Razborov and Rudich proved that any $\P/\poly$-natural property useful against $\P/\poly$ could break all strong pseudorandom generator candidates in $\P/\poly$. More generally, $\P/\poly$-natural properties useful against typical ${\cal C} \subset \P/\poly$ imply there are no strong pseudorandom functions in ${\cal C}$ (but such functions are believed to exist, even when ${\cal C}=\TC^0$~\cite{Naor-Reingold04}).

The natural property framework (as originally defined) only applies to strings encoding Boolean functions, with lengths always equal to a power of two. In this paper, we also consider the obvious extension of the natural property concept to \emph{arbitrary} length strings. We call such objects \emph{natural algorithms}, to emphasize that they are best viewed as algorithms operating on inputs of arbitrary length. For a string $x$ of length $n$, let $\ell$ be the smallest integer such that $2^{\ell} \geq n$. Recall we defined the \emph{Boolean function corresponding to $x$} to be $f_x : \{0,1\}^{\ell} \rightarrow \{0,1\}$  with truth table $x0^{2^{\ell}-n}$. 

\begin{definition} \label{natalg} A \emph{$\Gamma$-natural algorithm $A$ useful against ${\cal C}$} satisfies the axioms:
\begin{compactitem}
\item (Constructivity) \indent~~~\indent $L(A)$ is in $\Gamma$, 
\item (Largeness)  \indent~~~\indent For all $n$, $A$ accepts at least a $1/n^{O(1)}$ fraction of all $n$-bit strings,
\item (Usefulness)  \indent~~~\indent There are infinitely many $n$ such that \begin{compactitem}
\item[(a)] $A$ accepts at least one string $x$ of length $n$, and 
\item[(b)] for all $y$ of length $n$ accepted by $A$, the function $f_y$ does not have $n^k$-size ${\cal C}$ circuits. 
\end{compactitem}
\end{compactitem}
\end{definition}

The above definition of natural algorithm does not radically change the notion of usefulness (due to Lemma~\ref{truthtable} in Section~\ref{tt-ckt}); that is, padding a modest number of zeroes onto a string does not significantly alter the circuit complexity of the function represented by the string. However, the generalization to arbitrary input lengths is very useful for connecting the ideas of natural proofs to derandomization and circuit lower bounds. 

\subsection{Related Work}

\paragraph{Equivalences between algorithms \& lower bounds} Some of our results are equivalences between algorithm design problems and circuit lower bounds. Equivalences between derandomization hypotheses and circuit lower bounds have been known for some time, and recently there has been an increase in results of this form. Nisan and Wigderson~\cite{Nisan-Wigderson94} famously proved an equivalence between ``approximate'' circuit lower bounds and the existence of pseudorandom generators. 
Impagliazzo and Wigderson~\cite{Impagliazzo-Wigderson01} prove that $\BPP\neq\EXP$ implies deterministic subexponential-time \emph{heuristic} algorithms for $\BPP$ (the simulation succeeds on most inputs drawn from an efficiently samplable distribution, for infinitely many input lengths). As the opposite direction can be shown to hold, this is actually an equivalence. (Impagliazzo, Kabanets, and Wigderson~\cite{IKW} proved another such equivalence, which we discuss below.) Two more recent examples are Jansen and Santhanam~\cite{JS12}, who give an equivalence between nontrivial algorithms for polynomial identity testing and lower bounds for the algebraic version of $\NEXP$, and Aydinlioglu and Van Melkebeek~\cite{AvM12}, who give an equivalence between $\Sigma_2$-simulations of Arthur-Merlin games and circuit lower bounds for $\Sigma_2 \EXP$.

\paragraph{Almost-Natural Proofs} Philosophically related to the present work, Chow~\cite{Chow} showed that if strong pseudorandom generators do exist, then there \emph{is} a proof of $\NP \not\subset \P/\poly$ that is \emph{almost-natural}, where the fraction of inputs in the largeness condition is relaxed from $1/2^{O(n)}$ to $1/2^{n^{\poly(\log n)}}$. Hence the Natural Proofs barrier was already known to be sensitive to relaxations of  largeness. To compare, we show that removing the largeness condition entirely results in a direct \emph{equivalence} between the existence of ``almost-natural'' properties and circuit lower bounds against $\NEXP$. Chow also proved relevant unconditional results: for example, there exists a $\SIZE[O(n)]$-natural property that is $1/2^{n^{(\log n)^{\omega(1)}}}$-large and useful against $\P/\poly$. Theorem~\ref{equiv} shows that if $\SIZE[O(n)]$ could be replaced with $\P$, then $\NEXP \not\subset \P/\poly$ follows.

\paragraph{The work of Impagliazzo-Kabanets-Wigderson (IKW)} Impagliazzo, Kabanets, and Wigderson~\cite{IKW} proved a theorem similar to one direction of Theorem~\ref{equiv}, showing that an $\NP$-natural property (without largeness) useful against $\P/\poly$ implies $\NEXP \not\subset \P/\poly$. 
Allender~\cite{Allender01} proved that there is a (non-large) property computable in $\NP$ useful against $\P/\poly$ if and only if there is such a property in uniform $\AC^0$. Hence his equivalence implies, at least for ${\cal C}=\P/\poly$, that the ``polynomial-time'' guarantee of Theorem~\ref{equiv} can be relaxed to ``$\AC^0$.''  

IKW~\cite{IKW} also give an equivalence between $\NEXP$ lower bounds and an algorithmic problem: $\NEXP \not\subset \P/\poly$ if and only if 
the acceptance probability of any circuit can be approximated, for infinitely many circuit sizes, in nondeterministic subexponential time with subpolynomial advice. The major differences between their equivalence and Theorems~\ref{equiv} and~\ref{equiv2} are in the underlying computational problems and the algorithmic guarantees: they study subexponential-time algorithms for approximating acceptance probabilities of circuits, while we study algorithms which estimate the circuit complexities of given functions. Moreover, their equivalence is less general with respect to circuit classes; for example, it is not known how to prove an analogue of their equivalence for $\ACC$. 

Since they proved that the existence of $\NP$-natural properties useful against $\P/\poly$ imply that $\NEXP \not\subset \P/\poly$, IKW posed the interesting open problem: 
\begin{center}
\emph{Does the existence of a $\P$-natural property useful against $\P/\poly$ imply $\EXP \not\subset \P/\poly$?}
\end{center}
Our work shows that the \emph{absence} of a $\P$-natural property useful against $\P/\poly$ implies new lower bounds: 
\begin{claim} \label{NPZPP} If there is no $\P$-natural property useful against $\P/\poly$, then $\NP \neq \ZPP$.\end{claim} 
\begin{proof} We prove the contrapositive. If $\NP=\ZPP$, then there is  a $\ZPP$-natural property useful against $\P/\poly$ (since there are trivially $\coNP$-natural properties). Theorem~\ref{RPnaturalprop} implies that there is also a $\P$-natural property useful against $\P/\poly$. 
\end{proof}

Therefore, an affirmative answer to IKW's problem would prove that $\EXP \neq \ZPP$:
\begin{theorem} If ($\P$-natural properties useful against $\P/\poly$ $\Rightarrow$ $\EXP \not\subset \P/\poly$) is true, then $\EXP\neq\ZPP$ unconditionally.
\end{theorem}
\begin{proof}\begin{align*}
\text{We have }\EXP = \ZPP &\Rightarrow \NP = \ZPP\\
 &\Rightarrow \text{there are $\P$-natural properties useful against $\P/\poly$, by Claim~\ref{NPZPP}}\\
 & \Rightarrow \EXP \not\subset \P/\poly, \text{ by assumption}\\
 & \Rightarrow \EXP \neq \ZPP.
\end{align*} 
 Thus $\EXP \neq \ZPP$.
\end{proof}

\section{NEXP Lower Bounds and Useful Properties}\label{nexp}

In this section, we prove equivalences between $\NEXP$ circuit lower bounds and some relaxations of natural properties:

\begin{reminder}{Theorem~\ref{equiv}} For all typical ${\cal C}$, $\NEXP \not\subset {\cal C}$ if and only if there is a polynomial-time computable property of Boolean functions that is $\iouseful$ against ${\cal C}$ with $O(\log n)$ bits of advice.
\end{reminder}

\begin{reminder}{Theorem~\ref{equiv2}} For all typical ${\cal C}$, $\NEXP \not\subset {\cal C}$ if and only if there is a polynomial-time algorithm that is $\iouseful$ against ${\cal C}$.
\end{reminder}

Our proofs of these theorems take several steps (they could be shortened, as in Oliveira's survey~\cite{Oliveira13}, but the overall proofs would be less informative). First, we give an equivalence between the existence of small circuits for $\NEXP$ and the existence of small circuits encoding \emph{witnesses} to $\NEXP$ languages (Theorem~\ref{cktswitnesses}), strengthening results of Impagliazzo, Kabanets, and Wigderson~\cite{IKW} (who essentially proved one direction of the equivalence). Second, we prove an equivalence between the \emph{non-existence} of size-$s(O(n))$ witness circuits for $\NEXP$ and the existence of a $\P$-constructive property $P_s$ useful against size $s(O(n))$ circuits (Theorem~\ref{witnessequiv}), for all circuit sizes $s(n)$. For each polynomial $s(n)=n^k$, this yields a (potentially different) useful property $P_s$; to get a single property that works for all polynomial circuit sizes, we show that there exists a ``universal'' $\P$-constructive property $P^{\star}$: if for every circuit size $s$ there is \emph{some} $\P$-constructive useful property $P_s$, this particular property $P^{\star}$ is useful for all $s$ (Theorem~\ref{universalprop}). 

We first need a definition of what it means for a language (and a complexity class) to have small circuits encoding witnesses. We restrict ourselves to ``good'' verifiers which examine witnesses of length equal to a power of two, so that witnesses can be viewed as truth tables of Boolean functions:

\begin{definition}\label{witnessckts2} Let $L \in \NTIME[t(n)]$ where $t(n)\geq n$ is constructible, and let ${\cal C}$ be a circuit class. An algorithm $V(x,y)$ is a \emph{good predicate for $L$} if
\begin{compactitem}
\item $V$ runs in time $O(\poly(|y|+t(|x|)))$ and 
\item for all $x \in \{0,1\}^{\star}$, $x \in L$ if and only if there is a string $y$ such that $|y|=2^{\ell} \leq O(t(|x|))$ for some $\ell$ (a \emph{witness for $x$}) such that $V(x,y)$ accepts.
\end{compactitem}
Let $L(V)$ denote the language accepted by $V$. 
\end{definition}

For every $L \in \NTIME[t(n)]$, basic complexity arguments show that there is at least one good predicate $V$ such that $L = L(V)$. Furthermore, for every reasonable verifier $V$ used to define an $\NEXP$ language $L$, there is an equivalent good predicate $V'$ (with possibly slightly longer witness lengths). Now we define what it means for a verifier to have small-circuit witnesses:

\begin{definition} Let $V$ be a good predicate. \emph{$V$ has ${\cal C}$ witnesses of size $s(n)$} if for all strings $x$, if $x \in L$ then there is a ${\cal C}$-circuit $C_x$ of size at most $s(n)$ such that $V(x,tt(C_x(\cdot)))$ accepts. 

\emph{$L$ has ${\cal C}$ witnesses of $s(n)$ size} if for all good predicates $V$ for $L$, $V$ has ${\cal C}$ witnesses of size at most $s(n)$.\footnote{{\bf N.B.} For circuit classes ${\cal C}$ where the depth $d$ and/or modulus $m$ may be bounded, we also quantify this $d$ and $m$ simultaneously with the size parameter $s(n)$. That is, the depth, size, and modulus parameters are chosen prior to choosing an input, as usual.} 

The class $\NTIME[t(n)]$ \emph{has ${\cal C}$ witnesses of size $s(n)$} if for every language $L \in \NTIME[t(n)]$, $L$ has ${\cal C}$ witnesses of at most $s(n)$ size. The meaning of $\NEXP$ having ${\cal C}$ witnesses is defined analogously.
\end{definition}

The above definition of circuit witnesses allows, for every $x$, a different circuit $C_x$ encoding a witness for $x$. We will also consider a stronger notion of \emph{oblivious} witnesses, where a single circuit $C_n$ encodes witnesses for all $x \in L$ of length $n$.

\begin{definition}\label{witnessckts} Let $L \in \NTIME[t(n)]$, and let ${\cal C}$ be a circuit class.  
\emph{$L$ has oblivious ${\cal C}$ witnesses of size $s(n)$} if for every good predicate $V$ for $L$, there is a ${\cal C}$ circuit family $\{C_n\}$ of size $s(n)$ such that for all $x \in \{0,1\}^{\star}$, if $x \in L$ then $V(x,tt(C_{|x|}(x,\cdot))$ accepts.\footnote{That is, the truth table of $C_{|x|}$ with $x$ hard-coded is a valid witness for $x$.}

$\NTIME[t(n)]$ \emph{has oblivious ${\cal C}$ witnesses} if every $L \in \NTIME[t(n)]$ has oblivious ${\cal C}$ witnesses. The meaning of $\NEXP$ having ${\cal C}$ witnesses is defined analogously.
\end{definition}

We establish an equivalence between the existence of small circuits for $\NEXP$ and small circuits for $\NEXP$ witnesses, in both the oblivious and normal senses. 

\begin{theorem} \label{cktswitnesses} Let ${\cal C}$ be a typical polynomial-size circuit class. The following are equivalent:
\begin{compactenum}
\item[(1)] $\NEXP \subset {\cal C}$
\item[(2)] $\NEXP$ has ${\cal C}$ witnesses 
\item[(3)] $\NEXP$ has oblivious ${\cal C}$ witnesses 
\end{compactenum}
\end{theorem}

\begin{proof} $(1) \Rightarrow (2)$ Impagliazzo, Kabanets, and Wigderson~\cite{IKW} proved this direction for ${\cal C}=\P/\poly$. The other cases of ${\cal C}$ were observed in prior work~\cite{Williams10,Williams11}.

$(2) \Rightarrow (3)$ Assume $\NEXP$ has ${\cal C}$ witnesses (implicitly, they are of polynomial size). Let $V(x,y)$ be a good predicate for an $\NEXP$ problem that (without loss of generality) accepts witnesses $y$ of length exactly $2^{p(|x|)}$, for some polynomial $p(n)$. We will construct a ${\cal C}$-circuit family $\{C_n\}$ such that $x \in L$ if and only if $V(x,tt(C_{|x|}(x,\cdot)))$ accepts (recall $tt(C_{|x|}(x,\cdot))$ is the truth table of the circuit $C_{|x|}$ with $x$ hard-coded and the remaining inputs are free). The idea is to construct a new verifier that ``merges'' witnesses for all inputs of a given length into a single witness. (This theme will reappear throughout the paper.)

Let $x_1,\ldots,x_{2^n}$ be the list of strings of length $n$ in lexicographical order. We define a new good predicate $V'$ which takes a pair $(x,q)$ where $x \in \{0,1\}^n$ and $q = 0,\ldots,2^n$, along with $y$ of length $2^{n+p(n)}$:

\smallskip

\begin{framed}

$V'((x,q),y)$: \emph{Accept} if and only if
\begin{compactitem}
\item $y = b_1 z_1\cdots b_{2^{|x|}} z_{2^{|x|}}$, where for all $i=1,\ldots,2^{|x|}$, $b_i \in \{0,1\}$ and $z_i \in \{0,1\}^{2^{p(|x|)}}$, 
\item exactly $q$ of the $b_i$'s are $1$, 
\item for all $i$'s such that $b_i = 1$, $V(x_i,z_i)$ accepts.
\item for all $i$'s such that $b_i = 0$, $z_i = 0^{2^{p(|x|)}}$. 
\end{compactitem}

\end{framed}

\smallskip

$V'$ runs in time exponential in $|x|$; by assumption, $V'$ has ${\cal C}$ witnesses of polynomial size. Observe that the computation of $V'$ does not depend on the input $x$, only the length $|x|$.

To obtain oblivious ${\cal C}$ witnesses for $V$,
let $q_n$ be the actual number of $x$ of length $n$ such that $x \in L(V)$. Then for every $y''$ such that $V'((x',q_n),y'')$ accepts, the string $y''$ must encode a valid witnesses $z_i$ for every $x_i \in L(V)$. By assumption, there is a circuit $C_{(x',q_n)}$ such that $C_{(x',q_n)}(i)$ outputs the $i$th bit of $y''$. This circuit $C_{(x',q_n)}$ yields the desired witness circuit: indeed, the circuit $D_n(x,j) := C_{(x',q_n)}(x\circ j)$ (where $x \circ j$ denotes the concatenation of $x$ and $j$ as binary strings) prints the $j$th bit of a valid witness for $x$ (or it prints $0$, if $x \notin L(V)$).

$(3) \Rightarrow (1)$ Assume $\NEXP$ has oblivious ${\cal C}$ witnesses. Let $M$ be a nondeterministic exponential-time machine. We want to give a ${\cal C}$-circuit family recognizing $L(M)$. First, we define a good  predicate $V_k$:

\smallskip

\begin{framed}

$V_k(x,y)$: For all circuits $C$ of size $|x|^k+k$,\\
\indent \indent \indent \indent \indent \indent If $tt(C)$ encodes an accepting computation history of $M(x)$, then\\
\indent \indent \indent \indent \indent \indent ~~~~\emph{accept} if and only if the first bit of $y$ is $1$. \\
\indent \indent\indent \indent ~~~ End for \\
\indent \indent \indent \indent  ~~~ \emph{Accept} if and only if the first bit of $y$ is $0$.

\end{framed}

By assumption, there is a $k$ such that accepting computation histories of $M$ on all length $n$ inputs can be encoded with a single ${\cal C}$-circuit family of size at most $n^k+k$. For such a $k$, $V_k$ will run in $2^{O(n^k)}$ time and will always find a circuit $C$ encoding an accepting computation history of $M(x)$, when $x \in L(M)$. Therefore, $V_k(x,y)$ accepts if and only if \[[(\text{first bit of }y=1) \wedge (x \in L(M))] \vee [(\text{first bit of }y=0) \wedge (x \notin L(M))].\] Now, because $V_k$ is an good predicate for the $\NEXP$ language $L(M)$, we can apply the assumption again to $V_k$ itself, meaning there is a ${\cal C}$-circuit family $\{C_n\}$ encoding witnesses for $V_k$ obliviously. This family can be easily used to compute $L(M)$: define the circuit $D_n$ for $n$-bit instances of $L(M)$ to output the first bit of the witness encoded by $C_n(x,\cdot)$.
\end{proof}

Next, we prove a tight relation between witnesses for $\NE$ computations and constructive useful properties. (This equivalence will be useful for proving new consequences later.) Here, the typical circuit class ${\cal C}$ does not have to be polynomial-size bounded, and the size function $s(n)$ quantified below can be any reasonable function in the range $[n^2,2^n/(2n)]$ (for example). We have two versions of the relation: one for constructive properties of Boolean functions (defined only on $2^n$-bit strings) and one for polynomial-time algorithms (running on strings of all possible lengths).

\begin{theorem} \label{witnessequiv} For all size functions $s(n) \in [n^2,2^n/(2n)]$, the following are equivalent: 
\begin{enumerate}
\item There is a $c \in (0,1]$ such that $\NTIME[2^{O(n)}]$ does not have $s(cn)$ size witness circuits from ${\cal C}$.
\item There is a $c \in (0,1]$ and a $\P/(\log n)$-computable property of Boolean functions that is $\iouseful$ against ${\cal C}$-circuits of size at most $s(cn)$.\footnote{For circuit classes ${\cal C}$ with depth bound $d$, this $d$ will be universally quantified after $c$. So for example, there is a $c$ such that for all constant $d$, $\NTIME[2^{O(n)}]$ does not have $s(cn)$ size depth-$d$ $\AC^0[6]$ witnesses, if and only if there is a $c$ such that for all $d$, there is a polynomial-time algorithm $\iouseful$ against depth-$d$ $\AC^0[6]$ circuits of size $s(cn)$.}
\item There is a $c \in (0,1]$ and a polynomial-time algorithm that is $\iouseful$ against ${\cal C}$-circuits of size at most $s(cn)$.\end{enumerate}
\end{theorem}

\begin{proof} 
$(1) \Rightarrow (2)$  Suppose $\NTIME[2^{O(n)}]$ does not have $s(c \cdot n)$-size witness ${\cal C}$-circuits for some $c \in (0,1]$. Then there must be a good predicate $V$ running in $\TIME[2^{d n}]$ for some $d \geq 1$ that does not have $s(c \cdot n)$-size witnesses. Hence there is an infinite subsequence of ``bad'' inputs $\{x'_i\}$ such that for all $i$, $x'_i \in L(V)$, but for every $y$ such that $V(x'_i,y)$ accepts, $y$ requires $s(c \cdot |x'_i|)$ size ${\cal C}$-circuits to encode.

To give a $\P/(\log n)$-computable property of Boolean functions ${\cal P}$ that is $\iouseful$ against ${\cal C}$-circuits, 
 simply define ${\cal P}(f)$ with advice $x'_i$ to be true if and only if $f : \{0,1\}^{d |x'_i|} \rightarrow \{0,1\}$ and $V(x'_i,f)$ accepts (when $f$ is construed as a $2^{d |x_i|}$-bit string). The property ${\cal P}$ is clearly implementable in $\P/(\log n)$ (the advice can be anything when no appropriate $x'_i$ exists), and for infinitely many input lengths $\ell$, there is a string $x'_i \in L(V)$ of length $\ell$ such that every string $y$ of length $2^{d\ell}$ accepted by $V(x'_i,y)$ requires $s(c \cdot \ell)$ size ${\cal C}$-circuits as a Boolean function. Hence for infinitely many $\ell$, the property ${\cal P}$ is true of at least one Boolean function on $d \ell$ bits, and is false for all functions on $d \ell$ bits with $s(c \cdot \ell)$ size ${\cal C}$-circuits, for some fixed $d$.

(2) $\Rightarrow$ (3) Let ${\cal P}$ be a property of Boolean functions with $\log n$ bits of advice, implemented by a polynomial-time algorithm $B(\cdot,\cdot)$, which is $\iouseful$ against ${\cal C}$-circuits of size $s(cn)$. We give a polynomial-time \emph{algorithm} $A$ with no advice that is $\iouseful$ against ${\cal C}$-circuits of size at most $s(cn)$.  
Again, let $x_1,\ldots,x_{2^{\ell}}$ be the $\ell$-bit strings in lexicographical order in the following.

\begin{framed}
$A(y)$: If $y$ does not have the form $z01^k$, with $|z|=2^{\ell}$, for some $k=0,\ldots,2^{\ell}-1$ and $\ell$, then \emph{reject}.\\
\indent \indent ~~~~~~ Otherwise, compute $k$ by counting the trailing $1$'s at the end of $y$.\\ 
\indent \indent ~~~~~~ \emph{Accept} if and only if $B(z,x_k)$ accepts.
\end{framed}

Let $\ell$ be an integer such that the property ${\cal P}$, with the appropriate advice $x_k$ of length $d\ell$, is $\iouseful$ for functions on $\ell$ bits. Then for every $(z,x_k)$ pair accepted by the algorithm $B$, the Boolean function defined by $z$ of length $2^{\ell}$ is not computable with $s(c \cdot \ell)$-size ${\cal C}$-circuits. 

Observe that, for each $\ell$, and every possible $k=0,\ldots,2^{\ell}-1$, there is \emph{exactly} one input length, namely $n=2^{\ell}+k+1$, for which the input $x_k$ of length $\ell$ will be considered, along with all possible $z$'s of length $2^{\ell}$. Therefore, on those infinitely many input lengths $n$ for which the corresponding input $x_k$ of length $\ell$ equals some bad input $x'_j$, $A$ is useful against size-$s(c \cdot \ell)$ circuits from ${\cal C}$. 

$(3) \Rightarrow (1)$  Let $A$ be a $\poly(n)$-time algorithm that is $\iouseful$ against $s(c \cdot n)$-size  ${\cal C}$ circuits for some fixed constant $c$. In the following, let $x_k$ be the $k$th string in the lexicographical ordering of strings of length $|x_k|$. Define a machine:

\begin{framed}
$M(x_k,T)$: If $|T| \neq 2^{|x_k|}$, \emph{reject}. If $k > |T|/2$, \emph{reject}. \\
\indent \indent \indent ~~~~~~ Otherwise, strip the last $k-1$ bits from $T$, obtaining a string $T'$ of length $2^{|x_k|}-(k-1)$.\\
\indent \indent \indent ~~~~~~ \emph{Accept} if and only if $A(T')$ accepts. 
\end{framed}

Now define $L = \{x ~|~ (\exists~T:|T|=2^{|x|})[M(x,T)~\text{accepts}]\}$. Note that $L \in \NTIME[2^{O(n)}]$, and that $M$ is a good verifier for $L$. By our assumption that $A$ is a polytime useful algorithm, there are infinitely many integers $\ell$ such that
\begin{compactitem}
\item[(1)] $A$ accepts at least one string $y_{\ell}$ of length $\ell$, and 
\item[(2)] if $A$ accepts $y_{\ell}$ of length $\ell$, then the Boolean function corresponding to $y_{\ell}$ (possibly obtained by padding zeroes to the end of $y_{\ell}$) has circuit complexity greater than $s(c \cdot \ell)$. 
\end{compactitem}
For each such $\ell$, let $j_{\ell}$ be the smallest integer such that $2^{j_{\ell}} \geq \ell$. Define $i_{\ell} := 2^{j_{\ell}}-\ell$; that is, $i_{\ell} \in \{0,1,\ldots,2^{j_{\ell}-1}-1\}$ equals the number of zeroes needed to pad $y_{\ell}$ so that the length becomes a power of two. In the following, let $x_1,\ldots,x_{2^{j_{\ell}}}$ be the list of all $j_{\ell}$-bit strings in lexicographical order.

Then, $M(x_{i_{\ell}},T)$ accepts if and only if $|T| = 2^{j_\ell}$, $T = y_{\ell}z$ for some $z$ with $|z|=i_{\ell}$, and $A(y_{\ell})$ accepts. For infinitely many $\ell$, each such $y_{\ell}$ has the property that $y_{\ell}0^{i_{\ell}}$ has circuit complexity greater than $s(c \cdot j_{\ell})$, therefore each of the strings $T$ such that $M(x_{i_{\ell}},T)$ accepts must have circuit complexity greater than $s(c \cdot j_{\ell})-j_{\ell}^{1+o(1)}$ as well, by Proposition~\ref{truthtable}. So there is an infinite sequence of inputs $\{x'_{\ell}\}$ such that all strings $x'_{\ell}$ are in $L$, and all witnesses of $x'_{\ell}$ have circuit complexity greater than $s(c \cdot |x'_{\ell}|)-|x'_{\ell}|^{1+o(1)}$. Hence $L$ is a language in $\NTIME[2^{O(n)}]$ that does not have $(s(c \cdot n)-n^{1+o(1)})$-size witnesses. Since $s(n) \geq n^2$, we have completed the proof of this direction.
\end{proof}

Using complete languages for $\NEXP$, one can obtain an explicit property in $\P$ that is useful against ${\cal C}$ circuits, if there is \emph{any} constructive useful property. This universality means that, if there are multiple constructive properties that are useful against various circuit size functions, then there is one constructive property useful against all these size functions.

\begin{theorem}\label{universalprop}
Let $\{s_k(n)\}$ be an infinite family of functions such that for all $k$, there is a polynomial-time algorithm $P_k$ (or, polynomial-time property of Boolean functions with $\log n$ bits of advice) that is $\iouseful$ against all ${\cal C}$-circuits of $s_k(n)$ size. Then there is a single $\P$-computable algorithm $P^{\star}$ such that, for all $k$, there is a $c > 0$ such that $P^{\star}$ is $\iouseful$ against all ${\cal C}$-circuits of $s_k(cn)$ size.\footnote{For depth-bounded/modulus-bounded circuit classes ${\cal C}$, an analogous statement holds where we quantify not only over $k$ but also the depth $d$ and modulus $m$.}
\end{theorem}

\begin{proof} Let $b(n)$ denote the $n$th string of $\{0,1\}^{\star}$ in lexicographical order. The {\sc Succinct Halting} problem consists of all triples $\langle M,x,b(n)\rangle$ such that the nondeterministic TM $M$ accepts $x$ within at most $n$ steps. 
Define the algorithm

\smallskip

\begin{framed}

{\sc History}$(y)$: Compute $z = b(|y|)$. If $z$ does not have the form $\langle M,x,b(n)\rangle$, reject. \emph{Accept} if and only if there is a prefix $y'$ of $y$ with length equal to a power of two such that $y'$ encodes an accepting computation history to $z \in \text{\sc SuccinctHalting}$.

\end{framed}

\smallskip

Observe that {\sc History} is implementable in polynomial time. The theorem follows from the claim:

\begin{claim} {\sc History} is $\iouseful$ against ${\cal C}$ circuits of size $s(cn)$ for some $c > 0$ if and only if there is some polynomial-time algorithm (possibly with $\log n$ bits of advice) that is $\iouseful$ against ${\cal C}$ circuits of size $s(n)$.
\end{claim}

To see why Theorem~\ref{universalprop} follows, observe that if we have infinitely many properties $P_k$, each of which is $\iouseful$ against ${\cal C}$ circuits of $s_k(n)$ size, then for every $k$, {\sc History} will be useful against $s_k(n)$ size ${\cal C}$ circuits.

One direction of the claim is obvious. For the other, suppose there is a polynomial-time property with $\log n$ bits of advice (or a polynomial-time algorithm) $\iouseful$ against ${\cal C}$-circuits of size $s(n)$. By Theorem~\ref{witnessequiv}, $\NTIME[2^{O(n)}]$ does not have $s(dn)$ size witnesses from ${\cal C}$ for some constant $d$. Let $V$ be a good predicate running in time $2^{kn}$ that does not have $s(dn)$-size ${\cal C}$ witnesses, and let $M$ be the corresponding nondeterministic machine which, on $x$, guesses a $y$ and accepts iff $V(x,y)$ accepts. It follows that there are infinitely many instances of {\sc SuccinctHalting} of the form $\langle M,x,b(2^{k|x|})\rangle$ that do not have ${\cal C}$ witnesses of size $s(cn)$ for some constant $c$.
Therefore, there are infinitely many $z_i=\langle M_i,x_i,n_i\rangle$ in $\text{\sc SuccinctHalting}$, where every accepting computation history $y'$ of $M_i(x_i)$ has greater than $s(cn)$-size ${\cal C}$-circuit complexity. Then for all $n$ such that $z_i = b(n)$ for some $i$, there is a $y$ of length $n$ such that {\sc History}$(y)$ accepts but for all $y''$ which encode functions with ${\cal C}$-circuits of $s(cn)$-size, {\sc History}$(y'')$ rejects (by Proposition~\ref{truthtable0}; note $y''$ has length equal to a power of two). Hence {\sc History} is $\iouseful$ against ${\cal C}$ circuits of size $s(cn)$. This concludes the proof of the theorem.
\end{proof}

Putting it all together, we obtain Theorem~\ref{equiv} and Theorem~\ref{equiv2}:

\begin{proofof}{Theorem~\ref{equiv} and Theorem~\ref{equiv2}} We prove Theorem~\ref{equiv2}; the proof of Theorem~\ref{equiv} is analogous, and we add parenthetical remarks below about how to prove it. Let ${\cal C}$ be a typical class (of polynomial-size circuits). By Theorem~\ref{cktswitnesses}, we have $\NEXP \not\subset {\cal C}$ if and only if for every $k$, $\NEXP$ does not have ${\cal C}$ witnesses of $n^k$ size.

Setting $s(n)=n^k$ for arbitrary $k$ in Theorem~\ref{witnessequiv}, we infer that for every $k$, we have the equivalence: $\NEXP$ does not have ${\cal C}$ witnesses of $n^k$ size if and only if there is $c > 0$ and a polynomial-time algorithm that is $\iouseful$ against all ${\cal C}$-circuits of size at most $(cn)^k$. (Note that Theorem~\ref{witnessequiv} also implies an equivalence between the above two conditions and the existence of a $\P/(\log n)$-computable property useful against ${\cal C}$-circuits of size $(cn)^k$.) 

Applying Theorem~\ref{universalprop}, we conclude that $\NEXP \not\subset {\cal C}$ if and only if there is a polynomial-time algorithm such that, for all $k$, it is $\iouseful$ against all ${\cal C}$-circuits of size at most $n^k$. 
\end{proofof}

\section{New ACC Lower Bounds}

In this section, we prove new lower bounds against $\ACC$. Our approach uses a new nondeterministic simulation of randomized computation (assuming small circuits for $\ACC$). The simulation itself uses several ingredients. First, we prove an exponential-size lower bound on the sizes of $\ACC$ circuits encoding \emph{witnesses} for $\NTIME[2^{O(n)}]$. (Recall that, for $\NEXP$, the best known $\ACC$ size lower bounds are only ``third-exponential''~\cite{Williams11}.) Second, we use the connection between witness size lower bounds and constructive useful properties of Theorem~\ref{witnessequiv}. The third ingredient is a well-known hardness-randomness connection: from a constructive useful property, we can nondeterministically guess a hard function, verify its hardness using the property, then use the hard function to construct a pseudorandom generator. (Here, we will need to make an assumption like $\P \subset \ACC$, as it is not known how to convert hardness into pseudorandomness in the $\ACC$ setting~\cite{Shaltiel-Viola10}.)

\subsection{Exponential Lower Bounds for Encoding NEXP Witnesses}

\begin{reminder}{Theorem~\ref{expwitnessACC}} For all $d$, $m$ there is an $\eps = 1/m^{\Theta(d)}$ such that $\NTIME[2^{O(n)}]$ does not have $2^{n^{\eps}}$-size $d$-depth $\AC^0[m]$ witnesses.\footnote{The $m^{\Theta(d)}$ factor arises from the ACC-SAT algorithm in~\cite{Williams11}, which in turn comes from Beigel and Tarui's simulation of $\ACC$ in SYM-AND~\cite{Beigel-Tarui}.}
\end{reminder}

The proof is quite related in structure to the $\NEXP \not\subset \ACC$ proof, so we will merely sketch how it is different.

\begin{proof} (Sketch) Assume $\NTIME[2^{O(n)}]$ has $2^{n^{\eps}}$-size $\ACC$ witnesses, for all $\eps > 0$. We will show that the earlier framework~\cite{Williams11} can be adapted to still establish a contradiction. First, observe the assumption implies that $\TIME[2^{O(n)}]$ has $2^{n^{\eps}}$-size $\ACC$ circuits. (The proof is similar to the proof of Theorem~\ref{equiv}: for any given exponential-time algorithm $A$, one can set up a good predicate that only accepts its input of length $n$ if the witness is a truth table for the $2^n$-bit function computed by $A$ on $n$-bit inputs. Then, a witness circuit for this $x$ is a circuit for the entire function on $n$ bits.) Therefore (by Lemma 3.1 in~\cite{Williams11}) there is a nondeterministic $2^{n-n^{\delta}}$ time algorithm $A$ (where $\delta$ depends on the depth and modulus of $\ACC$ circuits for Circuit Evaluation) that, given any circuit $C$ of size $n^{O(1)}$ and $n$ inputs, generates an equivalent $\ACC$ circuit $C'$ of $2^{n^{\eps}}$ size, for all $\eps > 0$. (More precisely, there is some computation path on which $A$ generates such a circuit, and on every path, it either prints such a circuit or outputs \emph{fail}.)

The rest of the proof is analogous to prior $\NEXP$ lower bounds~\cite{Williams11}; we sketch the details for completeness. Our goal is to simulate every $L \in \NTIME[2^n]$ in nondeterministic time $2^{n-n^{\delta}}$, which will contradict the nondeterministic time hierarchy of {\v{Z}}\'{a}k~\cite{Zak83}. Given an instance $x$ of $L$, we first reduce $L$ to the $\NEXP$-complete {\sc Succinct 3SAT} problem using an efficient polynomial-time reduction. This yields an unrestricted circuit $D$ of size $n^{O(1)}$ and $n+O(\log n)$ inputs with truth table equal to a formula $F$, such that $F$ is satisfiable if and only if $x \in L$. We run algorithm $A$ on $D$ to obtain an equivalent $2^{n^{\eps}}$ size $\ACC$ circuit $D'$. Then we guess a $2^{n^{\eps}}$ size $\ACC$ circuit $E$ with truth table equal to a satisfying assignment for $F$. (If $x \in L$, then such a circuit exists, by assumption.) By combining copies of $D'$ and copies of $E$, we can obtain a single $\ACC$ circuit $C$ with $n+O(\log n)$ inputs which is unsatisfiable if and only if $E$ encodes a satisfying assignment for $F$. By calling a nontrivial satisfiability algorithm for ACC, we get a nondeterministic $2^{n-n^{\delta}}$ time simulation for every $L$, a contradiction.
\end{proof}

Applying Theorem~\ref{witnessequiv} and its corollary to the lower bound of Theorem~\ref{expwitnessACC}, we can conclude:

\begin{corollary}\label{accuseful} For all $d,m$, there is an $\eps = 1/m^{\Theta(d)}$ and a $\P$-computable property that is $\iouseful$ against all depth-$d$ $\AC^0[m]$ circuits of size at most $2^{n^{\eps}}$.
\end{corollary}

Hence there is an efficient way of distinguishing some functions from all functions computable with subexponential-size $\ACC$ circuits. Let CAPP be the problem: {\em given a circuit $C$, output $p \in [0,1]$ satisfying} \[|Pr_{x}[C(x)=1] - p| < 1/6.\] That is, we wish to approximate the acceptance probability of $C$ to within $1/6$. We can give a quasi-polynomial time nondeterministic algorithm for CAPP, assuming $\P$ is in quasi-polynomial size $\ACC$.

\begin{theorem}\label{derandomizeACC1} Suppose $\P$ has $\ACC$ circuits of size $n^{\log n}$. Then there is a constant $c$ such that for infinitely many sizes $s$, CAPP for size $s$ circuits is computable in nondeterministic $2^{(\log s)^c}$ time.
\end{theorem}

Theorem~\ref{derandomizeACC1} is a surprisingly strong consequence: given that $\NEXP \not\subset \ACC$, one would expect only a $2^{O(n^{\eps})}$-time algorithm for CAPP, with $n^{\eps}$ bits of advice. (Indeed, from the results of IKW~\cite{IKW} one can derive such an algorithm, assuming $\P \subseteq \ACC$.) 

Before proving Theorem~\ref{derandomizeACC1}, we first extend Theorem~\ref{expwitnessACC} a little bit.
Recall a \emph{unary language} is a subset of $\{1^n~|~n \in \N\} \subseteq \{0,1\}^{\star}$. The proof of Theorem~\ref{expwitnessACC} also has the following consequence:

\begin{corollary}\label{unary} If $\P$ has $\ACC$ circuits of $n^{\log n}$ size, then for all $d$, $m$ there is an $\eps$ such that there are unary languages in $\NTIME[2^n]$ without $2^{n^{\eps}}$-size $d$-depth $\AC^0[m]$ witnesses.
\end{corollary}

\begin{proof} The tight nondeterministic time hierarchy of {\v{Z}}\'{a}k~\cite{Zak83} holds also for unary languages. That is, there is a unary $L \in \NTIME[2^n]\setminus \NTIME[2^n/n^{10}]$. So assume (for a contradiction to this hierarchy) that all unary languages in $\NTIME[2^n]$ have $2^{n^{\eps}}$ size witnesses for every $\eps > 0$. This says that, for every good predicate $V$ for every unary language $L \in \NTIME[2^n]$, every $1^n \in L$ has a witness $y$ with $2^{n^{\eps}}$-size circuit complexity. Choose a predicate $V$ that reduces a given unary $L$ to a {\sc Succinct3SAT} instance, then checks that its witness is a SAT assignment to the instance; by assumption, such SAT assignments must have circuit complexity at most $2^{n^{\eps}}$, for almost all $n$. By guessing such a circuit and assuming $\P$ has $n^{\log n}$-size $\ACC$ circuits, the remainder of the proof of Theorem~\ref{expwitnessACC} goes through: the simulation of arbitrary $L$ in $\NTIME[2^{n-n^{\delta}}]$ works and yields the contradiction.
\end{proof}

Corollary~\ref{unary} allows us to strengthen Corollary~\ref{accuseful}, to yield a ``nondeterministically constructive'' and useful property against $\ACC$. Informally, having a \emph{unary} language without small witness circuits allows us to obtain a derandomization \emph{without} advice, as there is no need to store a ``hard'' input for a given input length. In particular, the unconditional lower bound of Corollary~\ref{unary} can be used to build an efficient ``hardness test'' for $\ACC$ circuit complexity, which is then used with a pseudorandom generator to solve CAPP by guessing a hard function and verifying it with the test. This basic idea seems to have originated with~\cite{Kabanets-Cai00,Kabanets01}.

\begin{proofof}{Theorem~\ref{derandomizeACC1}} First we claim that, if $\P$ has $n^{\log n}$ size $\ACC$ circuits, then there is a $d^{\star}$ and $m^{\star}$ such that every Boolean function $f$ with unrestricted circuits of size $S$ has depth-$d^{\star}$ $\AC^0[m^{\star}]$ circuits of size at most $S^{\log S}$. To see this, consider the {\sc Circuit Evaluation} problem: {\em given a circuit $C$ and an input $x$, does $C(x)=1$?} Assuming $\P$ is in $n^{\log n}$ $\ACC$, this problem has a depth-$d^{\star}$ $AC^0[m^{\star}]$ circuit family $\{D_n\}$ of $n^{\log n}$ size, for some fixed $d^{\star}$ and $m^{\star}$. Therefore, by plugging in the description of any circuit $C$ of size $S$ into the input of the appropriate $\ACC$ circuit $D_{O(S)}$, we get an $\ACC$ circuit of fixed modulus and depth that is equivalent to $C$ and has size $O(S^{\log S})$.

By Corollary~\ref{unary}, there is an $\eps$ and a unary $L$ in $\NTIME[2^n]$ that does not have $2^{n^{\eps}}$ size $\AC^0[m^{\star}]$ witnesses of depth $d^{\star}$. By the previous paragraph (and assuming $\P$ is in $n^{\log n}$-size $\ACC$), it follows that $L$ does not have witnesses encoded with $2^{n^{\eps/2}}$-size unrestricted circuits. (Letting $S^{\log S} = 2^{n^{\eps}}$, we find that $S = 2^{n^{\eps/2}}$.) Let $V$ be a good predicate for $L$ that lacks such witnesses, and let $g$ be the constant in the pseudorandom generator of Theorem~\ref{hardness-randomness}. Consider the nondeterministic algorithm $P$ which, on input $1^s$, sets $n = (g\log s)^{2/\eps}$, guesses a string $Y$ of $2^n$ length, and outputs $Y$ if $V(1^n,Y)$ accepts (otherwise, $P$ outputs \emph{reject}). For infinitely many $s$, $P(1^s)$ nondeterministically generates strings $Y$ of $2^{(g\log s)^{2/\eps}}$ length that do not have $s^g = 2^{n^{\eps/2}}$ size circuits: as there is an infinite set of $\{n_i\}$ such that all witnesses to $1^{n_i}$ have circuit complexity at least $2^{(n_i)^{\eps/2}}$, there is an infinite set $\{s_i\}$ such that $P(1^{s_i})$ computes $n_i = (g\log s_i)^{2/\eps}$ and generates $Y$ which does not have $(s_i)^g = 2^{(n_i)^{\eps/2}}$ size circuits.

Given a circuit $C$ of size $s$, our nondeterministic simulation runs $P$ to generate $Y$. (If $P$ rejects, the simulation \emph{rejects}.) Applying Theorem~\ref{hardness-randomness}, $Y$ can be used to construct a $\poly(|Y|)$-time PRG $G(Y,\cdot) : \{0,1\}^{g\log |Y|} \rightarrow \{0,1\}^{s}$ which fools circuits of size $s$. By trying all $|Y|^g \leq 2^{O((\log s)^{2/\eps})}$ inputs to $G_Y$, we can approximate the acceptance probability of a size-$s$ circuit in $2^{O((\log s)^{2/\eps})}$ time. As $\eps$ depended only on $d^{\star}$ and $m^{\star}$, which are both constants, we can set $c = 3/\eps$ to complete the proof.
\end{proofof}

\subsection{A Slightly Stronger ACC Lower Bound}

Now we turn to proving lower bounds for the classes $\NE \cap \io\coNE$ and $\NE/1 \cap \coNE/1$. We will need an implication between circuits and Merlin-Arthur simulations that extends Babai-Fortnow-Nisan-Wigderson~\cite{Babai-Fortnow-Nisan-Wigderson93}:

\begin{theorem}[Lemma 8,~\cite{MiltersenVW99}] \label{MVW} Let $g(n) > 2^n$ and $s(n) \geq n$ be increasing and time constructible. There is a constant $c > 1$ such that $\TIME[2^{O(n)}] \subseteq \SIZE[s(n)] \Longrightarrow \TIME[g(n)] \subseteq {\sf MATIME}[s(3 \log g(n))^c]$.
\end{theorem}

That is, if we assume exponential time has $s(n)$-size circuits, we can simulate even larger time bounds with Merlin-Arthur games. This follows from the proof of $\EXP \subset \P/\poly \Longrightarrow \EXP = \MA$~(\cite{Babai-Fortnow-Nisan-Wigderson93}) combined with a padding argument.

\begin{reminder}{Theorem~\ref{NEcapcoNE}}
$\NE \cap \io\coNE$ and $\NE/1 \cap \coNE/1$ do not have $\ACC$ circuits of $n^{\log n}$ size.
\end{reminder}

\begin{proof}
Suppose $\NE \cap \io\coNE$ has $n^{\log n}$-size $\ACC$ circuits. We wish to derive a contradiction. Of course the assumption implies that $\TIME[2^{O(n)}]$ has $n^{\log n}$-size circuits as well. Applying Theorem~\ref{MVW} with $g(n)=2^{n^{2\log n}}$ and $s(n) = n^{\log n}$, we have \[\TIME[2^{n^{2\log n}}] \subseteq \MATIME[n^{O(\log^3 n)}].\] By Theorem~\ref{derandomizeACC1} and assuming that $\P$ has $\ACC$ circuits of size $n^{\log n}$, there is a constant $c$ and a pseudorandom generator with the following properties: for infinitely many circuit sizes $s$, the generator nondeterministically guesses a string $Y$ of length $2^{(\log s)^{c}}$, verifies $Y$ in $\poly(|Y|)$ deterministic time with a $\iouseful$ property $P$, then uses $Y$ to construct a PRG that runs in $\poly(|Y|)$ time deterministically over $\poly(|Y|)$ different seeds. The $\poly(|Y|)$ outputs of length $s$ can then be used to correctly approximate the acceptance probability of any size $s$ circuit.

We can use this generator to fool Merlin-Arthur games on infinitely many circuit sizes, as well as co-Merlin-Arthur games. Take a $n^{O(\log^3 n)}$-size circuit $C$ encoding the predicate in a given Merlin-Arthur game of that length ($C$ takes an input $x$, Merlin's string of length $n^{O(\log^3 n)}$, and Arthur's string of length $n^{O(\log^3 n)}$, and outputs a bit).
Our simulation first guesses Merlin's string $m$, then runs the PRG which guesses a $Y$ and verifies that $Y$ is a hard function; if the verification fails, we \emph{reject}.  
Then the simulation uses the PRG on $C(x,m,\cdot)$ to simulate Arthur's string and the final outcome, accepting if and only if the majority of strings generated by the PRG lead to acceptance. 
On infinitely many input lengths, 
the simulation of the Merlin-Arthur game will be ``faithful'' in the sense that the PRG simulating Arthur will work as intended. 

Hence there is a constant $d$ such that \begin{equation}\label{ncontain1}\TIME[2^{n^{2\log n}}] \subseteq \MATIME[n^{O(\log^3 n)}] \subseteq \io\NTIME[n^{\log^{d} n}].\end{equation}
As $\TIME[2^{n^{2\log n}}]$ is closed under complement, an analogous argument (applied to any machine accepting the complement of a given $\TIME[2^{n^{2\log n}}]$ language) implies
\begin{equation}\label{ncontain2}\TIME[2^{n^{2\log n}}] \subseteq {\sf coMATIME}[n^{O(\log^3 n)}] \subseteq \io\coNTIME[n^{\log^{d} n}].\end{equation} 

Let us look at these simulations more closely. Given a language $L$ in time $2^{n^{2\log n}}$ time, by \eqref{ncontain1} we have a language $L' \in \NTIME[n^{\log^{d} n}]$ which agrees with $L$ on infinitely many input lengths $n_1,n_2,\ldots$. Since $\TIME[2^{n^{2\log n}}]$ is closed under complement, for the language $\overline{L}$ (the complement of $L$) there is also a language $L'' \in \NTIME[n^{\log^d n}]$ which agrees with $\overline{L}$ on the same list of input lengths $n_1,n_2,\ldots$. Since $L'$ agrees with the complement of $L''$ on these infinitely many input lengths, we have that $L' \in \io\coNTIME[n^{\log^d n}]$, and therefore \[\TIME[2^{n^{2\log n}}] \subseteq \io(\NTIME \cap \io\coNTIME)[n^{\log^d n}].\] Assuming every language in $\NE \cap \io\coNE$ has circuits of size $n^{\log n}$, it follows that every language in the class $\io(\NE \cap \io\coNE)$ has circuits of size $n^{\log n}$ for infinitely many input lengths. Therefore \[\TIME[2^{n^{2\log n}}] \subset \io\SIZE[n^{\log n}].\] But this is a contradiction: for almost every $n$, by simply enumerating all $n^{\log n}$-size circuits and their $2^n$-bit truth tables, we can compute the lexicographically first Boolean function on $n$ bits which does not have $n^{\log n}$ size circuits, in $O(2^{n^{2\log n}})$ time. 

To prove a lower bound $\NE/1 \cap \coNE/1$, we follow precisely the same argument up to \eqref{ncontain1}, and make the following modifications. By using a bit of advice $y_n \in \{0,1\}$ to encode whether or not the PRG will be successful for a given input length $n$, we can simulate an arbitrary $L \in \TIME[2^{n^{2\log n}}]$ infinitely often in $\NE/1 \cap \coNE/1$. In particular, we define a nondeterministic $N$ and co-nondeterministic $N'$ which take an advice bit, as follows: if the advice bit is $0$, both simulations reject; otherwise, $N$ attempts to run the Merlin-Arthur simulation of $L$ (and $N'$ attempts to Merlin-Arthur simulate $\overline{L}$, respectively) as described above. When the advice bits are assigned appropriately on all input lengths, $N$ and $N'$ accept a language $L' \in \NE/1 \cap \coNE/1$ such that for all $n$, either $L' \cap \{0,1\}^n = \varnothing$ (for input lengths where the advice is set to $0$) or $L' \cap \{0,1\}^n = L \cap \{0,1\}^n$ (for infinitely many $n$). Therefore \[\TIME[2^{n^{2\log n}}]\subseteq \io(\NE/1 \cap \coNE/1),\] and the remainder of the argument concludes as above.
\end{proof}

We conclude the section by sketching how the above argument can be recast in a more generic form, as a connection between SAT algorithms and circuit lower bounds:

\begin{reminder}{Theorem~\ref{generic}} Let ${\cal C}$ be typical. Suppose the satisfiability problem for $n^{O(\log^c n)}$-size ${\cal C}$ circuits can be solved in $O(2^n/n^{10})$ time, for all constants $c$. Then $\NE \cap \io\coNE$ and $\NE/1 \cap \coNE/1$ do not have $n^{\log n}$-size ${\cal C}$ circuits.
\end{reminder}

\begin{proof} (Sketch) Suppose satisfiability for ${\cal C}$ circuits of $n^{O(\log^c n)}$ size is in $O(2^n/n^{10})$ time (for all $c$), and that $\NE \cap \io\coNE$ has $n^{\log n}$ size circuits.
By the proof of Theorem~\ref{derandomizeACC1}, assuming $\P$ has $n^{\log n}$ size ${\cal C}$ circuits, for all $\eps > 0$, we obtain a nondeterminstic algorithm $N$ running in $2^{2^{O(\log^{\eps} s)}}$ time on all circuits of size $s$ (for infinitely many $s$) and outputs a good approximation to the given circuit's acceptance probability. (In particular, from the assumptions we can derive a unary language computable in $\NTIME[2^n]$ that does not have witness circuits of $n^{\log^c n}$ size, for every $c$; this can be used to obtain a nondeterministic algorithm $N$ as in Theorem~\ref{derandomizeACC1}, by setting $s = n^{O(\log^c n)}$, solving for $n=2^{O((\log s)^{1/(c+1)})}$, then running the nondeterministic algorithm $N$ in $2^{O(n)} \leq 2^{2^{O(\log^{\eps} s)}}$ time, where $\eps \leq 1/(c+1)$.)

By the same argument as in the proof of Theorem~\ref{NEcapcoNE}, we obtain \[\TIME[2^{n^{2 \log n}}] \subseteq ({\sf MATIME} \cap {\sf coMATIME})[n^{O(\log^3 n)}].\] By applying algorithm $N$ to circuits of size $s = n^{O(\log^3 n)}$ and setting $\eps \ll 1/4$, we obtain \[({\sf MATIME} \cap {\sf coMATIME})[n^{O(\log^3 n)}] \subseteq \io(\NTIME \cap \io\coNTIME)[2^{O(n)}].\] But the latter class is in $\io\SIZE[n^{\log n}]$ by assumption; we conclude a contradiction as in Theorem~\ref{NEcapcoNE}.

Similarly as in Theorem~\ref{NEcapcoNE}, assuming $\NE/1 \cap \coNE/1$ has $n^{\log n}$ size circuits, we can conclude \[({\sf MATIME} \cap {\sf coMATIME})[n^{O(\log^3 n)}] \subseteq \io(\NTIME[2^{O(n)}]/1 \cap \coNTIME[2^{O(n)}]/1) \subset \SIZE[n^{\log n}],\] yielding another contradiction.
\end{proof}

\begin{reminder}{Theorem~\ref{generic2}} Suppose we can approximate the acceptance probability of any given $n^{O(\log^c n)}$-size circuit (with fan-in two and arbitrary depth) on $n$ inputs to within $1/6$, for all $c$, in $O(2^n/n^{10})$ time (even nondeterministically). Then $\NE \cap \io\coNE$ and $\NE/1 \cap \coNE/1$ do not have $n^{\log n}$-size circuits.
\end{reminder}

\begin{proof} (Sketch) For all the lower bound arguments given in this section, an algorithm which can approximate the acceptance probability of a given $n^{O(\log^c n)}$-size circuit can be applied in place of a faster SAT algorithm (\cite{Williams10,Williams11,SanthanamWilliams13}). That is, from the hypothesis of the theorem we can derive exponential-size witness circuit lower bounds for $\NEXP$ (as in Theorem~\ref{expwitnessACC}) and  infinitely-often correct pseudorandom generators against general circuits (as in Theorem~\ref{derandomizeACC1}). Therefore the proofs of Theorem~\ref{NEcapcoNE} and consequently Theorem~\ref{generic} also carry over under the hypothesis of the theorem.
\end{proof}

\section{Natural Properties and Derandomization}\label{natural}

In this section, we characterize (the nonexistence of) natural properties as a particular sort of derandomization problem, and exhibit several consequences. 

Let $\ZPE = \ZPTIME[2^{O(n)}]$, i.e., the class of languages solvable in $2^{O(n)}$ time with randomness and no error (the machine can output {\bf ?}, or \emph{don't know}). $\RE=\RTIME[2^{O(n)}]$ is its one-sided-error equivalent. Analogously to Definition~\ref{witnessckts2}, we define a witness notion for $\ZPE$ as follows:

\begin{definition} Let $L \in \ZPE$. A \emph{$\ZPE$ predicate for $L$} is a procedure $M(x,y)$ that runs in time $2^{O(|x|)}$ on inputs $y$ of length $2^{c|x|}$ for some constant $c$, such that for every $x$ and $y$,
\begin{compactitem}
\item The output of $M(x,y)$ is in the set $\{1,0,\text{\bf ?}\}$.
\item $x \in L \Longrightarrow \Pr_{y \in \{0,1\}^{2^{c|x|}}}[M(x,y)\text{ outputs }1] \geq 2/3$, and for all $y$ of length $2^{c|x|}$, $M(x,y)\in\{1,\text{\bf ?}\}$.
\item $x \notin L \Longrightarrow \Pr_{y \in \{0,1\}^{2^{c|x|}}}[M(x,y)\text{ outputs }0] \geq 2/3$, and for all $y$ of length $2^{c|x|}$, $M(x,y) \in \{0,\text{\bf ?}\}$.
\end{compactitem}

$\ZPE$ \emph{has ${\cal C}$ seeds} if for every $\ZPE$ predicate $M$, there is a $k$ such that for all $x$, there is a ${\cal C}$-circuit $C_x$ of size at most $|x|^k+k$ such that $M(x,tt(C_x)) \neq \text{\bf ?}$.\footnote{For circuit classes where the depth $d$ and/or modulus $m$ may be bounded, we also quantify this $d$ and $m$ simultaneously with the size parameter $k$. That is, the depth, size, and modulus parameters are chosen prior to choosing the circuit family, as usual.}

$\ZPE$ \emph{has ${\cal C}$ seeds for infinitely many input lengths} if for every $\ZPE$ predicate $M$, there is a $k$ such that for infinitely many $n$ and for all $x$ of length $n$, there is a ${\cal C}$-circuit $C_x$ of size at most $n^k+k$ such that $M(x,tt(C_x)) \neq \text{\bf ?}$.
\end{definition}

That is, ${\cal C}$ seeds for $\ZPE$ are succinct encodings of strings that lead to a decision by the algorithm. Analogously, we can define \emph{$\RE$ predicates} and the notion of \emph{$\RE$ having ${\cal C}$ seeds}: $\RE$ predicates will accept with probability at least $2/3$ when $x \in L$, but reject with probability $1$ when $x \notin L$. Hence, when $\RE$ has ${\cal C}$ seeds, we only require $x \in L$ to have small circuits $C_x$ encoding witnesses.

Succinct seeds for zero-error computation are closely related to uniform natural properties, as follows:

\begin{reminder}{Theorem~\ref{naturalequiv}} Let ${\cal C}$ be a typical polynomial-size circuit class. The following are equivalent:
\begin{compactenum}
\item There are no $\P$-natural properties $\iouseful$ (respectively, ae-useful\footnote{Here, \emph{ae-useful} is just the ``almost-everywhere useful'' version, where the property is required to be distinguish random functions from easy ones on almost every input length.}) against ${\cal C}$
\item $\ZPE$ has ${\cal C}$ seeds for almost all (resp., infinitely many) input lengths.
\end{compactenum}
\end{reminder}

The intuition is that, given a $\P$-natural useful property, its probability of acceptance can be amplified (at a mild cost to usefulness), yielding a $\ZPE$ predicate which accepts random strings with decent probability but still lacks small seeds. In the other direction, suppose a $\ZPE$ predicate has ``bad'' inputs that can't be decided using small circuits encoding seeds. This implies that a ``hitting set'' of exponential-length strings, sufficient for deciding all inputs of a given length, must have high circuit complexity---otherwise, all strings in the set would have low circuit complexity (by Lemma~\ref{truthtable}), but at least one such string decides even a bad input. Checking for a hitting set is then a $\P$-natural, useful property.

\begin{proofof}{Theorem~\ref{naturalequiv}} $(\neg(1) \Rightarrow \neg(2))$ Suppose there is a $\P$-natural property which is ae-useful (resp., $\iouseful$) against ${\cal C}$.  For some $c, d \geq 1$, this is an $n^c$-time algorithm $A$ such that, for almost all $n$ (resp., infinitely many $n$), $A$ accepts at least a $1/2^{d\log n} = 1/n^d$ fraction of $n$-bit inputs, for $n = 2^{\ell}$, and for almost all $n = 2^{\ell}$ (resp., for infinitely many $n$) and all $c$, $A$ rejects all $n$-bit inputs representing truth tables of $(\log n)^c$-size ${\cal C}$-circuits.

Let $b(n)$ denote the $n$th string in lexicographical order. Let $\eps > 0$ be sufficiently small. Define an algorithm $V$:

\smallskip

\begin{framed}

$V(x,z)$: If $x \neq b(|z|)$ then output {\bf ?}. If $|z| \neq 2^{(d+1)k+1}$ for some $k$, then output {\bf ?}.\\ 
\indent \indent ~~~ Partition $z$ into $t = 2^{dk+1}$ strings $z_1,\ldots,z_t$ each of length $2^{k}$.\\
\indent \indent ~~~  If $A(z_i)$ accepts for some $i$, then output $1$; else, output {\bf ?}.

\end{framed}

We claim $V$ is a $\ZPE$ predicate for $L = \{0,1\}^{\star}$. Consider a $z$ chosen at random. All $z_i$ of length $2^{dk}$ are independent random variables, and by assumption, $A$ accepts at least $1/|z_i|^d = 1/2^{dk}$ strings of that length. The probability that all $z_i$ are among the $(1-1/2^{dk})$ fraction of strings of length $2^{k}$ rejected by $A$ is at most $(1-1/2^{dk})^{2^{dk+1}} \leq \exp(-2) < 1/3$. Therefore $V$ accepts a random $z$ of the appropriate length with at least $2/3$ probability.

By construction, $V$ accepts $(x,z)$ precisely when $x=b(|z|)$ and some $z_i$ of length $2^{dk}$ is accepted by $A$. Hence for almost all $k$ (resp., infinitely many $k$), when $V(x,z)$ accepts on $z$ of length $2^{(d+1)k+1}$, some substring $z_i$ of length $2^k$ has ${\cal C}$-circuit complexity at least $(\log 2^{k})^c \geq \Omega(\log^c |z|)$. Therefore by Lemma~\ref{truthtable}, $z$ itself has ${\cal C}$-circuit complexity at least $\Omega((\log |z|)^c-(\log|z|)^{1+o(1)})$. As this holds for every $c$, the predicate $V$ does not have ${\cal C}$ seeds infinitely often (respectively, almost everywhere).

$(\neg(2) \Rightarrow \neg(1))$  Suppose there is a $\ZPE$ predicate $V$ that does not have ${\cal C}$ seeds almost everywhere (resp, infinitely often). This means that, for all $k$ and for infinitely many (resp., almost all) input lengths $n_i$, there is some input $x$ of length $n_i$ such that, for every string $r$ of length $2^{c n_i}$ satisfying $V(x,r) \neq \text{\bf ?}$, the ${\cal C}$-circuit complexity of $r$ is at least $(c n_i)^k$. (Note that the constant $c$ depends only on $V$.) Define a new predicate $V'$ as follows, intended to be executed on the inputs $x$ with lengths in $\{n_i\}$:

\begin{framed}
$V'(x,r)$: If $|r| \neq 2^{\ell + c|x|}$ where $\ell$ is the smallest integer such that $2|x| \leq 2^{\ell}$, \emph{reject}.\\
\indent \indent ~~~  Partition $r$ into $2^{\ell}$ strings $\{r_i\}$ of length $2^{c|x|}$ each.\\
\indent \indent ~~~ \emph{Accept} if and only if $V(x,r_i) \neq {\bf ?}$ for some $i$.
\end{framed}

For those inputs $x$ of length $n_i$, any $r$ accepted by $V'(x,r)$ does not have circuits of size $n_i^k$, due to Proposition~\ref{truthtable0} and the fact that such an $r$ contains a substring $r_i$ such that $V(x,r_i)$ accepts, hence $r_i$ has circuit complexity at least $n_i^k$. By standard probabilistic arguments and our choice of $\ell$, it is likely that the string $r$ encodes a hitting set for all inputs of length $n_i$, i.e., \[\Pr_{r \in \{0,1\}^{2^{cn_i + \ell}}}\left[(\exists x \in \{0,1\}^{n_i})(\forall~i=1,\ldots,2^{\ell})[V(x,r_i) = \text{\bf ?}]\right] < 1/3.\] Therefore, a randomly chosen $r$ of length $2^{cn_i+\ell}$ is accepted by $V'$, with probability at least $2/3$. Equipped with this knowledge, we now define an algorithm $A$ that defines a $\P$-natural property of Boolean functions:

\begin{framed}
$A(f)$: Given a Boolean function $f : \{0,1\}^{\ell'} \rightarrow \{0,1\}$, \\
\indent \indent ~~~  Compute the largest $n \in \Z$ such that $\ell' \geq \ell + cn$, \\
\indent \indent ~~~~~~~~~~~~~~~~ where $\ell = O(\log n)$ is the smallest integer satisfying $2n \leq 2^{\ell}$.\\
\indent \indent ~~~ Set $r$ to be the first $2^{\ell + cn}$ bits of $f$.\\
\indent \indent ~~~ Search over all strings in $\{0,1\}^n$ for an $x$ such that $V'(x,r) \neq \text{\bf ?}$ for some $r_i$. \\
\indent \indent ~~~ Output \emph{accept} if such an $x$ is found, otherwise \emph{reject}.
\end{framed}

The algorithm $A$ runs in $\poly(2^{\ell'})$ time, and accepts at least $1/2$ of its inputs. Furthermore, when the integer $n$ computed by $A$ is in the sequence $\{n_i\}$, $A$ rejects all $f$ with ${\cal C}$-circuit complexity at most $n_i^k = \Theta((\ell')^k)$: if $f$ had such circuits, then all substrings $r_i$ of $f$ would as well, by Proposition~\ref{truthtable0}. As this is true for every constant $k$, $A$ is a $\P$-natural property useful against polynomial-size ${\cal C}$ circuits.
\end{proofof}

To prove a related result for $\RE$ predicates, we first need a little more notation. Let $V$ be an $\RTIME[2^{kn}]$ predicate accepting a language $L$. For a given input length $n$, a set $S_n \subseteq \{0,1\}^{2^{kn}}$ is a \emph{hitting set for $V$ on $n$} if, for all $x \in L$ of length $n$, there is a $y \in S_n$ such that $V(x_n,y)$ accepts. For a string $T$ of length $m \cdot 2^{kn}$, $T$ \emph{encodes a hitting set for $V$ on $n$} if, breaking $T$ into $m$ strings $y_1,\ldots,y_m$ of equal length, the set $\{y_1,\ldots,y_m\}$ is a hitting set for $V$ on $n$.

We consider yet another relaxation of naturalness. For a typical circuit class ${\cal C}$, we say that a polynomial-time algorithm $A$ \emph{is io-$\P$-natural against ${\cal C}$} provided that, for every $k$ and infinitely many integers $n$,
\begin{compactitem}
\item  $A$ accepts at least a $1/\poly(n)$ fraction of $n$-bit inputs, and 
\item $A$ rejects all $n$-bit inputs $x$ such that the corresponding Boolean function $f_x$ has $((\log n)^k+k)$-size ${\cal C}$-circuits.\footnote{As usual, if ${\cal C}$ is also characterized by a depth $d$ or modulus constraint $m$, those $d$ and $m$ are quantified alongside $k$.} 
\end{compactitem}
(Compare with Definition~\ref{natalg}.) In the usual notion of natural properties, we are restricted to inputs with length equal to a power of two, and largeness holds almost everywhere; here, neither conditions are required. 

We can relate succinctly encoded hitting sets to natural algorithms as follows:

\begin{theorem} \label{rtimenatural} Suppose for all $c$, $\RTIME[2^{O(n)}]$ does not have $n^2$-size hitting sets encoded by $n^c$-size circuits. Then for all $c$, there is an io-$\P$-natural algorithm useful against $n^c$ size circuits.
\end{theorem}

\begin{proof}
The hypothesis says that for every $c$, there is an $\RTIME[2^{O(n)}]$ predicate $V_c$ accepting some language $L$ with the following property: for every $n^{c}$-size circuit family $\{C_n\}$, there are infinitely many $n$ where $tt(C_n)$ does not encode an $n^2$-size hitting set for $V_c$ on $n$.

We may obtain an io-natural algorithm computable in $\poly(N)$ time with $O(\log N)$ bits of advice (where $N$ is the length of the input), as follows. 

\begin{framed}
$A(Y,a)$: Given $Y$ of length $N=2^{kn + 2\log n}$,\\
\indent \indent ~~~ View the $O(\log N)$-bit advice string $a$ as the number of inputs of length $n$ in $L(V_c)$. \\
\indent \indent ~~~ Partition $Y$ into $y_1,\ldots,y_{2^{2\log n}}$ of length $2^{kn}$.\\
\indent \indent ~~~ Let $b$ be the number of $x$ of length $n\leq (\log N)/k$ such that $V_c(x,y_i)$ accepts for some $i$.\\
\indent \indent ~~~ If $(a=b)$ then \emph{accept} else \emph{reject}. 
\end{framed}

For infinitely many $N$, this procedure (with the appropriate advice string $a$) accepts a random string with high probability, because a random collection of $n^2$ strings is a hitting set, whp. On those same input lengths $N$, the procedure $A$ also rejects strings encoded by $n^{c}$-size circuit families, by assumption. Therefore $A$ defines io-$\P/(\log n)$-natural algorithm $A$ useful against $n^c$-size circuits, running on strings $Y$ with length equal to a power of two. 

We can use $A$ to design an io-$\P$-natural algorithm $A'$ that runs on arbitrary length strings, analogously to one direction of Theorem~\ref{witnessequiv}. For every $n \in \N$, we associate the interval $I_n = [n^2,(n+1)^2-1]$; note that the collection of $I_n$ is a partition of $\N$. Our algorithm $A'$ runs as follows:

\begin{framed}
$A'(X)$: On input $X$ of length $m$, determine $n$ such that $m \in I_n$. If $n$ is not a power of two, \emph{reject}.\\
\indent \indent ~~~ Compute $a = m-n^2$, and treat $a$ as a binary string of length $O(\log n)$.\\
\indent \indent ~~~ Let $Y$ be the first $n$ bits of $X$.\\
\indent \indent ~~~ Run $A(Y,a)$ and output the answer. 
\end{framed}

Observe that $A'(X)$ runs in $\poly(m)$ time. Since $a$ as defined in $A'$ is contained in $\{0,\ldots,2n\}$, the number $a$ can be treated as an advice string of length $(\log n)$ for $n$-bit inputs.

For infinitely many input lengths $n_i$, the original algorithm $A$ (equipped with the appropriate advice $a_i$) satisfies largeness and usefulness against $n^c$-size circuits. For each such $n_i$, there is an slightly larger input length $m_i$ such that the number of $n_i$-bit inputs in $L(V_c)$ is exactly $a_i = m_i - n_i^2$. 

On these integers $m_i$, the algorithm $A'(\cdot)$ also satisfies largeness and usefulness, since it is essentially equivalent to running $A(\cdot,a_i)$ on inputs of length $n_i$. More precisely, since the input length has increased by a square ($m_i = \Theta(n_i^2)$), the strings of length $m_i$ define functions on only twice as many input bits as $n_i$. Therefore, when $A(x,a_i)$ accepts (hence $x$ has circuit complexity at least $(\log n_i)^c$), by Lemma~\ref{truthtable} we may conclude that the original input $X$ to $A'$ defines a Boolean function on at most $2\log m_i \leq 4\log n_i$ bits, with circuit complexity at least $(\log n_i)^c - (\log n_i)^{1+o(1)}$. Therefore the new algorithm $A'$ is $\iouseful$ against circuits of size up to $(n/4)^c$. As this condition holds for every constant $c$, the theorem follows.
\end{proof}

The other direction (from io-$\P$-natural algorithms to $\RTIME[2^{O(n)}]$) seems difficult to satisfy: it could be that, for infinitely many $n$, the natural algorithm does not obey any nice promise conditions on the number of accepted inputs of length $n$.

\subsection{Unconditional Mild Derandomizations}

We are now prepared to give some unconditionally-true derandomization results. The first one is:

\begin{reminder}{Theorem~\ref{derand}} Either $\RTIME[2^{O(n)}] \subset \SIZE[n^c]$ for some $c$, or $\BPP \subset \io\ZPTIME[2^{n^{\eps}}]/n^{\eps}$ for all $\eps > 0$.
\end{reminder}

To give intuition for the proof, we compare with the ``easy witness'' method of Kabanets~\cite{Kabanets01}, which shows that $\RP$ can be \emph{pseudo}-simulated in $\io\ZPTIME[2^{n^{\eps}}]$ (no efficient adversary can generate an input on which the simulation fails, almost everywhere). That simulation works as follows: for all $\eps > 0$, given an $\RP$ predicate, try all $n^{\eps}$-size circuits and check if any encode a good seed for the predicate. If this always happens (against all efficient adversaries), then we can simulate $\RP$ in subexponential time. Otherwise, some efficient algorithm can generate, infinitely often, inputs on which this simulation fails. This algorithm generates the truth table of a function that does not have $n^{\eps}$-size circuits; this hard function can be used to derandomize $\BPP$.

In order to get a nontrivial simulation that works on all inputs for many lengths, we consider \emph{easy hitting sets}: sets of strings (as in Theorem~\ref{rtimenatural}) that contain seeds for \emph{all} inputs of a given length, encoded by $n^c$-size circuits (where $c$ does not have to be tiny, but rather a fixed constant). When such seeds exist for some $c$, we can use $\tilde{O}(n^c)$ bits of advice to simulate $\RP$ deterministically. Otherwise, we apply Theorem~\ref{rtimenatural} to obtain an io-$\P$-natural algorithm which can be used (by randomly guessing a hard function) to simulate $\BPP$ in subexponential time. This allows us to avoid explicit enumeration of all small circuits; instead, we let the circuit size exceed the input length, and enumerate over (short) inputs in our natural property.

\begin{proofof}{Theorem~\ref{derand}} First, suppose there is a $c \geq 2$ so that for every $\RTIME[2^{O(n)}]$ predicate $V$ accepting a language $L$, there is an $n^{c-1}$-size circuit family $\{C_n\}$ such that for almost all $n$, $C_n$ has $O(n)$ inputs and its truth table encodes a hitting set for $V$ on $n$ with $2^{2 \log n}$ strings. That is, the truth table of $C_n$ is a string $Y$ of length $\ell = 2^{2\log n} \cdot 2^{kn}$ for a constant $k$, with the property that when we break $Y$ into $O(n^2)$ equal length strings $y_1,\ldots,y_{2^{2\log n}}$, the set $\{y_i\}$ is a hitting set for $V$ on $n$. Then it follows immediately that $\RTIME[2^{O(n)}] \subset \TIME[2^{O(n)}]/n^c$, because for almost all lengths $n$, we can provide the appropriate $n^{c-1}$-size circuit $C_n$ as $O(n^c)$ bits of advice, and recognize $L$ on any $n$-bit input $x$ by evaluating $C$ on all its possible inputs, testing the resulting hitting set of $O(n^2)$ size with $x$. (We will show later how to strengthen this case.)

If the above supposition is false, that means for every $c$, there is an $\RTIME[2^{O(n)}]$ predicate $V_c$ accepting some language $L$ with the following property: for every $n^{c}$-size circuit family $\{C_n\}$, there are infinitely many $n$ such that the truth table of $C_n$ does not encode a hitting set for $V$ on $n$. Theorem~\ref{rtimenatural} says that for all $c$, we can extract an io-$\P$-natural algorithm $A_c$ useful against $n^c$ size circuits, for all $c$. In particular, the proof of Theorem~\ref{rtimenatural} shows that for all $c$ there are infinitely many $n$ and  $m \in [2^{n/3},2^{3n}]$ such that $A_c$ is useful and large on its inputs of length $m$. So if we want a function $f : \{0,1\}^{O(n)} \rightarrow \{0,1\}$ that does not have $n^k$ size circuits, then by setting $c=k$, providing the number $m$ as $O(n)$ bits of advice, and randomly selecting $Y$ of $m$ bits, we can generate an $f$ that has guaranteed high circuit complexity, with zero error.

For every $k$, we can simulate any language in $\BPTIME[O(n^k)]$ (two-sided randomized $n^k$ time), as follows. Given any $k$ and $\eps > 0$, set $c=gk/\eps$ (where $g$ is the constant in Theorem~\ref{hardness-randomness}). On input $x$ of length $n$, our ZP simulation will have hard-coded advice of length $O(n^{\eps})$, specifying an input length $m = 2^{\Theta(n^{\eps})}$. Then it chooses a random string $Y$ of length $m$, and computes $A_c(Y)$. If $A_c(Y)$ rejects, then the simulation outputs \emph{don't know}. (For the proper advice $m$ and the proper input lengths, this case will happen with low probability.) Otherwise, for infinitely many $n$, $Y$ is an $m = 2^{\Theta(n^{\eps})}$ bit string with circuit complexity at least $(n^{\eps})^c \geq n^{gk}$.

Applying Theorem~\ref{hardness-randomness}, $Y$ can be used to construct a PRG $G_Y : \{0,1\}^{g \log |Y|} \rightarrow \{0,1\}^{n^{3k}}$ which fools circuits of size $n^{3k}$, where $d$ is a universal constant (independent of $\eps$ and $k$). Each call to $G_Y$ takes $\poly(|Y|) \leq 2^{O(n^{\eps})}$ time. Trying all $|Y|^g \leq 2^{O(n^{\eps})}$ seeds to $G_Y$, we can approximate the acceptance probability of a $n^{3k}$-size circuit simulating any $\BPTIME[O(n^k)]$ language on $n$-bit inputs, thereby determining acceptance/rejection of any $n$-bit input.

Now we have either (1) $\RTIME[2^{O(n)}] \subset \TIME[2^{O(n)}]/n^c$ for some $c$, or (2) $\BPP \subset \io\ZPTIME[2^{n^{\eps}}]/n^{\eps}$ for all $\eps > 0$. To complete the proof, we recall that Babai-Fortnow-Nisan-Wigderson~\cite{Babai-Fortnow-Nisan-Wigderson93} proved that if $\BPP \not\subset \io\SUBEXP$ then $\EXP \subset \P/\poly$. Therefore, if case (2) does not hold, the first case can be improved: using a complete language for $\E$, we infer from $\EXP \subset \P/\poly$ that $\TIME[2^{O(n)}] \subset \SIZE[n^c]$ for some $c$, so $\RTIME[2^{O(n)}] \subset \SIZE[n^c]$ for some constant $c$.
\end{proofof}

\begin{reminder}{Corollary~\ref{rp}} For some constant $c$, $\RP \subseteq \io\ZPSUBEXP/n^c$.
\end{reminder}

\begin{proof} By Theorem~\ref{derand}, there are two cases: (1) $\RTIME[2^{O(n)}] \subset \SIZE[n^c]$ for some $c$, or (2) $\BPP \subset \io\ZPTIME[2^{n^{\eps}}]/n^{\eps}$ for all $\eps$. In case (1), $\RP \subseteq \RTIME[2^{O(n)}] \subseteq \TIME[n^c]/n^c$. In case (2), $\RP \subseteq \BPP \subseteq \io\ZPTIME[2^{n^{\eps}}]/n^{\eps}$.
\end{proof}

The simulation can be ported over to Arthur-Merlin games. Recall that a language $L$ is in $\AM$ if and only if there is a $k$ and deterministic algorithm $V(x,y,z)$ running in time $|x|^k$ with the properties:
\begin{itemize}
\item If $x \in L$ then $\Pr_{y \in \{0,1\}^{|x|^k}}[\exists z \in \{0,1\}^{|x|^k}~V(x,y,z)\text{~accepts}] = 1$.
\item If $x \notin L$ then $\Pr_{y \in \{0,1\}^{|x|^k}}[\forall z \in \{0,1\}^{|x|^k}~V(x,y,z)\text{~rejects}] > 2/3$.
\end{itemize}
An $\AM$ computation corresponds to an interaction between a randomized verifier (Arthur) that sends random string $y$, and a prover (Merlin) that nondeterminstically guesses a string $z$. 

\begin{reminder}{Corollary~\ref{am}}
 For some $c \geq 1$, $\AM \subseteq \io\Sigma_2 \SUBEXP/n^c$.
\end{reminder}

The problem of finding nontrivial relationships between $\AM$ and $\Sigma_2 \P$ has been open for some time~\cite{GST03,AvM12}. 

\begin{proof} (Sketch) The proof is roughly analogous to relativizing Theorem~\ref{derand} with an $\NP$ oracle; for completeness, we include some of the details. Instead of hitting sets for $\RP$ computations, we consider hitting sets for $\AM$ computations: a $\poly(n)$-size set $S$ of $n^k$-bit strings that can replace the role of $y$ (Arthur) in the $\AM$ computation. (Such hitting sets always exist, by a probabilistic argument.) That is, on all strings $x$ of length $n$, computing the probability of $(\exists z)[V(x,y,z)]$ over all $y \in S$ allows us to approximate the probability over \emph{all} $n^k$-bit strings.  Instead of considering hitting sets that are succinctly encoded by typical circuits, we consider $\AM$ hitting sets that are succinctly encoded by circuits with oracle gates that compute SAT. There are two possible cases: 

1. \emph{There is a $c$ such that for all languages $L \in \AM$ and verifiers $V_c$ for $L$, there is an $n^c$-size SAT-oracle circuit family encoding hitting sets for $V_c$, on almost all input lengths $n$.} In this case, we can put $\AM$ in the class $\P^{\NP}/\tilde{O}(n^c)$: we can use $\tilde{O}(n^c)$ advice to store a circuit encoding a hitting set for each input length $n$, evaluate this circuit on $n^{O(1)}$ inputs in $\P^{\NP}$, producing the hitting set, then use the hitting set and the $\NP$ oracle to simulate the $\AM$ computation.

2. \emph{For all $c$, there is some verifier $V$ of some $\AM$ language such that, for infinitely many input lengths $n$, every hitting set for $V$ over all inputs of length $n$ has SAT-oracle circuit complexity greater than $n^c$.} First we show how to use this case to check that a given string $Y$ has high SAT-oracle circuit complexity for infinitely many input lengths; the argument is similar to prior ones. Given a string $Y$, let $k \geq 1$ be a parameter, let $\eps > 0$ be sufficiently small, and consider the verifier $V_{10k/\eps}$ on all inputs of length $n = m^{\eps}$ (where $n$ is one of the infinitely many input lengths which are ``good''). We can verify that the string $Y$ encodes a hitting set for $V_{10k/\eps}$ on inputs of length $n$, as follows. First we guess which of the $2^n$ strings of length $n$ are accepted, and which are rejected (comparing our guesses against the $O(n)$ bits of advice, which will encode the total number of accepted inputs of length $n$). For each string that is guessed to be accepted, we use the set $S$ and nondeterminism to simulate Arthur and Merlin's acceptance in $2^n \cdot \poly(n)$ time. Then for each string  that is guessed to be rejected, we use the string $Y$ and universal guessing to confirm that Arthur and Merlin reject in $2^n \cdot \poly(n)$ time. This is a $\Sigma_2$ computation running in time $2^{O(n)} \leq 2^{O(m^{\eps})}$, which (when given the appropriate advice of length $O(m^{\eps})$) correctly determines that at least some string $Y$ has SAT-oracle circuit complexity at least $(m^{\eps})^{10k/\eps} \geq n^{10k}$, on infinitely many input lengths.

Now suppose we want to simulate an $\AM$ computation on inputs of length $m$ running in time $m^k$. Then we can simulate the $\AM$ computation in $\io\Sigma_2 \TIME[2^{n^{\eps}}]/O(n^{\eps})$, as follows: we guess a string $Y$ with high SAT-oracle circuit complexity, and apply known relativizing results in derandomization (in particular Theorems~3.2 and 3.3 from \cite{Klivans-vanMelkebeek02}) that use the string $Y$ to simulate $\AM$ computations in $\NSUBEXP$.  Then we apply the aforementioned $\Sigma_2$ procedure to verify that the $Y$ guessed has high SAT-oracle circuit complexity. We accept if and only if the simulation of $\AM$ accepts and the verification of $Y$ accepts.
\end{proof}

It looks plausible that Corollary~\ref{am} could be combined with other results (for example, the work on lower bounds against fixed-polynomial advice, of Buhrman-Fortnow-Santhanam~\cite{BFS09}) to separate $\Sigma_2 \EXP$ from $\AM$.

Another application of Theorem~\ref{derand} is an unexpected equivalence between the infamous separation problem $\NEXP \neq \BPP$ and zero-error simulations of $\BPP$. We need one more definition: ${\sf Heuristic~}{\cal C}$ is the class of languages $L$ such that there is a $L' \in {\cal C}$ whereby, for almost every $n$, the symmetric difference $(L \cap \{0,1\}^n) \Delta (L' \cap \{0,1\}^n)$ has cardinality less than $2^n/n$.\footnote{{\bf N.B.} This is a weaker definition than usually stated, but it will suffice for our purposes.} (That is, there is a language in ${\cal C}$ that ``agrees'' with $L$ on at least a $1-1/n$ fraction of inputs.) The infinitely often version $\io{\sf Heuristic~}{\cal C}$ is defined analogously.

\begin{reminder}{Theorem~\ref{NEXPBPPequiv}}
$\NEXP \neq \BPP$ if and only if for all $\eps > 0$, $\BPP \subseteq \io{\sf Heuristic}\ZPTIME[2^{n^{\eps}}]/n^{\eps}$.
\end{reminder}

This extends an amazing result of Impagliazzo and Wigderson~\cite{Impagliazzo-Wigderson01} that $\EXP \neq \BPP$ if and only if for all $\eps > 0$, $\BPP \subseteq \io{\sf Heuristic}\TIME[2^{n^{\eps}}]$. It is interesting that $\NEXP$ versus $\BPP$, a problem concerning the power of nondeterminism, is equivalent to a statement about derandomization of $\BPP$ \emph{without} nondeterminism. Theorem~\ref{NEXPBPPequiv} should also be contrasted with the $\NEXP$ vs $\P/\poly$ equivalence of IKW~\cite{IKW}: $\NEXP \not\subset \P/\poly$ if and only if $\MA \subseteq \io\NTIME[2^{n^{\eps}}]/n^{\eps}$, for all $\eps > 0$.

\begin{proofof}{Theorem~\ref{NEXPBPPequiv}}
First, assume $\BPP$ is not in $\io{\sf Heuristic}\ZPTIME[2^{n^{\eps}}]/n^{\eps}$ for some $\eps$. Then $\BPP \not\subseteq \io\ZPTIME[2^{n^{\eps}}]/n^{\eps}$, so by Theorem~\ref{derand} we have that $\RTIME[2^{O(n)}]$ has size-$n^c$ seeds, which implies $\REXP = \EXP$. The hypothesis also implies that $\BPP$ is not in $\io{\sf Heuristic}\TIME[2^{n^{\eps}}]$, so by Impagliazzo and Wigderson~\cite{Impagliazzo-Wigderson01} we have $\EXP = \BPP$. Therefore $\REXP = \BPP$. But this implies $\NP \subseteq \BPP$, so by Ko's theorem~\cite{Ko82} we have $\NP = \RP$. Finally, by padding, $\NEXP = \REXP = \BPP$.

For the other direction, suppose $\NEXP = \BPP$ and $\BPP \subseteq \io{\sf Heuristic}\ZPTIME[2^{n^{\eps}}]/n^{\eps}$ for all $\eps > 0$. We wish to prove a contradiction. The two assumptions together say that $\NEXP \subseteq \io{\sf Heuristic}\NTIME[2^{n^{\eps}}]/n^{\eps}$ for all $\eps > 0$. $\NEXP=\BPP$ implies $\NEXP = \EXP$, and since $\NE$ has a linear-time complete language, we have $\NTIME[2^{O(n)}] \subseteq \TIME[2^{O(n^c)}]$ for some constant $c$. (More precisely, the {\sc SuccinctHalting} problem from Theorem~\ref{equiv} can be solved in $2^{O(n^c)}$ time for some $c$, and every language in $\NTIME[2^{O(n)}]$ can be reduced in linear time to {\sc SuccinctHalting}.) As a consequence, \begin{equation}\label{EXPcontain}\EXP = \NEXP \subseteq \bigcap_{\eps > 0}\io{\sf Heuristic}\NTIME[2^{n^{\eps}}]/n^{\eps} \subseteq \bigcap_{\eps > 0}\io{\sf Heuristic}\TIME[2^{O(n^c)}]/n^{\eps}.\end{equation} The last inclusion in \eqref{EXPcontain} can be proved as follows: let $L \in \bigcap_{\eps > 0}\io{\sf Heuristic}\NTIME[2^{n^{\eps}}]/n^{\eps}$ be arbitrary, and let $L' \in \bigcap_{\eps > 0} \NTIME[2^{n^{\eps}}]/n^{\eps}$ be such that $(L \cap \{0,1\}^n) \Delta (L' \cap \{0,1\}^n) \leq 2^n/n$ on infinitely many $n$. This means that, for any $\eps$, $L'$ can be solved using a collection of nondeterministic machines $\{M_n\}$ running in $2^{n^{\eps}}$ time such that $M_n$ solves all instances on $n$ bits and the description of $M_n$ can be encoded in $O(n^{\eps})$ bits. To get a collection of equivalent deterministic machines, let $M_n$ be the advice for inputs of length $n$; on any input $x$ of length $n$, call the $2^{O(n^c)}$ time algorithm for {\sc SuccinctHalting} on the input $\langle M_n,x,b(2^{n^{\eps}})\rangle$, where $b(m)$ is the binary encoding of $m$. Using standard encodings, this instance has $n+O(n^{\eps})$ length, hence it is solved deterministically in $2^{O(n^c)}$ time.

Finally, we prove that the above inclusion \eqref{EXPcontain} is false, by direct diagonalization. That is, we can find an $L \in \EXP$ such that $L \not\in \io{\sf Heuristic}\TIME[2^{O(n^c)}]/n^{1/2}$. Let $\{M_i\}$ be a list of all $2^{n^c}$ time machines. We will give  a $2^{n^{c+1}}$-time $M$ diagonalizing (even heuristically) against all $\{M_i\}$ with $n^{1/2}$ advice. For every $n$, $M$ divides up its $n$-bit inputs into blocks of length $B = 1+n^{1/2}+\log n$, with $2^n/B$ blocks in total. On input $x$ of length $n$, $M$ identifies the block containing $x$, letting $x_1,\ldots,x_{B}$ be the strings in the that block. Let $\{a_j\}$ be the set of all possible advice strings of length $n^{1/2}$. The following loop is performed:

Let $S_0 = \{(j,k)~|~j=1,\ldots,n,~k=1,\ldots,2^{n^{1/2}}\}$.
For $i=1,\ldots,B$, decide that $M$ accepts $x_i$ iff the majority of $M_j(x_i,a_k)$ reject over all $(j,k)\in S_{i-1}$. Set $S_i$ to be the subset of $S_{i-1}$ containing those $(M_j,a_k)$ which agree with $M$ on $x_i$. If $x_i=x$ then output the decision.

Observe that $M$ runs in $B \cdot n \cdot 2^{O(n^c)} \leq O(2^{n^{c+1}})$ time. For every block and every $i$, we have $|S_i| \leq |S_{i-1}|/2$. Since $|S_0|=2^{n^{1/2}}\cdot n$, this implies that $|S_{B}|=0$. So for every block, every pair $(M_j,a_k)$ disagrees with $M$ on at least one input. Therefore every pair $(M_j,a_k)$ disagrees with $M$ on at least $2^n/B > 2^n/n$ inputs, one from each block, and this happens for almost all input lengths $n$. Summing up, for almost every $n$ we have that $M$ disagrees with every $M_i$ and its $n^{1/2}$ bits of advice, on greater than a $1/n$ fraction of $n$-bit inputs. That is, $L(M) \in \EXP$ but $L(M) \not\in \io{\sf Heuristic}\TIME[2^{O(n^c)}]/n^{1/2}$.
\end{proofof}

\begin{remark} An anonymous reviewer observed that the above proof, very slightly modified, also shows that $\NEXP \neq \BPP$ if and only if for all $\eps > 0$, $\BPP \subseteq \io{\sf Heuristic}\NTIME[2^{n^{\eps}}]/n^{\eps}$. That is, separating $\NEXP$ from $\BPP$ is \emph{equivalent} to obtaining a nontrivial simulation of $\BPP$ with nondeterminism.
\end{remark}

\section{Unconditional Derandomization of Natural Properties}
\label{RPnatural}

In this last section, we show how one can use similar ideas to generically ``derandomize'' natural properties, in the sense that $\RP$-natural properties entail $\P$-natural ones. The formal claim is:

\begin{reminder}{Theorem~\ref{RPnaturalprop}} If there exists a $\RP$-natural property $P$ useful against a class ${\cal C}$, then there exists a  $\P$-natural property $P'$ useful against ${\cal C}$.
\end{reminder}

That is, suppose there is a randomized algorithm that can distinguish hard functions from easy functions with one-sided error---the algorithm may err on some hard functions, but never on any easy functions. Then we can obtain a deterministic algorithm with essentially the same functionality. The idea behind $P'$ is directly inspired by other arguments in the paper (such as the proof of Theorem~\ref{naturalequiv}): we split the input string $T$ into small substrings, and feed the substrings as inputs to $P$ while the whole input string $T$ is used as randomness to $P$. 

\begin{proof} Suppose $A$ is a randomized polytime algorithm taking $n$ bits of input and $n^{k-2}$ bits of randomness (for some $k \geq 3$), deciding a large and useful property against $n^c$-size circuits for every $c$. For concreteness, let us say that $A$ accepts some $1/n^b$-fraction of $n$-bit inputs with probability at least $2/3$, and rejects all $n$-bit truth tables of $(\log n)^c$-size circuits, where $b \geq k$ (making $b$ larger is only a weaker guarantee). Standard amplification techniques show that, by increasing the randomness from $n^{k-2}$ to $n^k$, we can boost the success probability of $A$ to greater than $1-1/4^n$. 

Our deterministic algorithm $A'$ will, on $n$-bit input $T$, partition $T$ into substrings $T_1,\ldots,T_{n^{1-1/k}}$ of length at most $n^{1/k}$ each, and \emph{accept} if and only if $A(T_i,T)$ accepts for some $i$.

First, we show that $A'$ satisfies largeness. Consider the set $R$ of $n$-bit strings $T$ such that for all $n^{1/k}$-bit strings $x$, $A(x,T)$ accepts if and only if $A(x,T')$ accepts for some $n$-bit $T'$. As there are only $2^{n^{1/k}}$ strings on $n^{1/k}$ bits, and the probability that a random $n$-bit $T$ works for a given $n^{1/k}$-bit string is at least $1-1/4^{n^{1/k}}$, we have (by a union bound) that $|R| \geq 2^n\cdot(1-2^{n^{1/k}}/4^{n^{1/k}}) \geq 2^n\cdot(1-1/2^{n^{1/k}})$.

Now consider the set $S$ of all $n$-bit strings $T = T_1\cdots T_{n^{1-1/k}}$ (where for all $i$, $|T_i|=n^{1/k}$) such that $A(T_i,T')$ accepts for some $i$ and some $n$-bit $T'$. Since there are at least $t=2^{n^{1/k}}/n^{b/k}$ such strings $T_i$ of length $n^{1/k}$ (by largeness of $A$), the cardinality of $S$ is at least \[ n^{1-1/k}\cdot t\cdot\left(2^{n^{1/k}}-t\right)^{n^{1-1/k}-1} = n^{1-1/k}\cdot \frac{2^{n^{1/k}}}{n^{b/k}}\cdot\left(2^{n-n^{1/k}}\right)\cdot\left(1-1/n^{b/k}\right)^{n^{1-1/k}-1},\] as this expression just counts the number of strings $T$ with exactly one $T_i$ from the $t$ strings accepted by $A$. Since $b \geq k$, $(1-1/n^{b/k})^{n^{1-1/k}-1} \geq 1/e$, and the above expression simplifies to $\Omega(2^n/n^{1/k-1+b/k})$. Therefore, there is a constant $e = b/k+1/k-1$ such that $|S| \geq \Omega(2^n/n^e)$.

Observe that, if $T \in S \cap R$, then $A(T_i,T)$ accepts for some $i$ (where $T_i$ is defined as above). 
Applying the inequality $|S \cap R| \geq |S|+|R|-2^n$, there are at least $2^n(1/n^e - 1/2^{n^{1/k}})$ strings such that $A(T_i,T)$ accepts for some $i$. This is at least $2^n/n^{e+1}$ for sufficiently large $n$, so $A'$ satisfies largeness.

Second, we show that $A'$ is useful. Suppose for a contradiction that $A'(T)$ accepts for some $T$ with $(\log |T|)^c$ size circuits, where $c$ is an arbitrarily large (but fixed) constant. Then $A(T_i,T)$ must accept for some $i$. Because $A$ is useful against $n^d$-size circuits for all $d$, it must be that $T_i$ cannot have $(\log |T_i|)^{c+1}$ size circuits. However, recall that if a string $T$ has $(\log |T|)^c$ size circuits, then by Lemma~\ref{truthtable}, every $|T|^{1/k}$-length substring $T_i$ of $T$ has circuit complexity at most $(\log|T|)^c + (\log|T|)^{1+o(1)} \leq 2\cdot (k \cdot \log |T_i|)^{c}$. As $k$ is a fixed constant, this quantity is less than $(\log |T_i|)^{c+1}$ when $|T_i|$ is sufficiently large, a contradiction.
\end{proof}

\section{Conclusion}

Ketan Mulmuley has recently suggested that ``$\P \neq \NP$ because $\P$ is big, not because $\P$ is small''~\cite{Ketan11}. That is to say, the power of efficient computation is the true reason we can prove lower bounds. The equivalence in Theorems~\ref{equiv} and \ref{equiv2} between $\NEXP$ lower bounds and constructive useful properties can be viewed as one rigorous formalization of this intuition. We conclude with some open questions of interest.

\smallskip

$\bullet$ \emph{Do $\NEXP$ problems have witnesses that are {\bf average-case hard} for $\ACC$?} More precisely, are there $\NEXP$ predicates with the property that, for almost all valid witnesses of length $2^{O(n)}$, their corresponding Boolean functions on $O(n)$ variables are such that that no $\ACC$ circuit of polynomial size agrees with these functions on $1/2+1/\poly(n)$ of the inputs? Such predicates could be used to yield unconditional derandomized simulations of $\ACC$ circuits (using nondeterminism). The primary technical impediment seems to be that we do not think $\ACC$ can compute the Majority function, which appears to be necessary for hardness amplification (see~\cite{Shaltiel-Viola10}). But this should make it \emph{easier} to prove lower bounds against $\ACC$, not harder!

\smallskip

$\bullet$ \emph{Equivalences for non-uniform natural properties?} In this paper, we have mainly studied natural properties decidable by algorithms with $\log n$ bits of advice or less; however, the more general notion of $\P/\poly$-natural proofs has also been considered. Are there reasonable equivalences that can be derived between the existence of such properties, and lower bounds? 

\smallskip

$\bullet$ \emph{What algorithms follow from stronger lower bound assumptions?} There is an interesting tension between the assumptions ``$\NEXP \not\subset \P/\poly$'' and ``integer factorization is not in subexponential time.'' The first asserts nontrivial efficient algorithms for recognizing some hard Boolean functions (as seen in Theorems~\ref{equiv} and \ref{equiv2}); the second denies efficient algorithms for recognizing a non-negligible fraction of hard Boolean functions~\cite{Kabanets-Cai00,AllenderBKMR06}. An equivalence involving $\NP \not\subset \P/\poly$ could yield more powerful algorithms for recognizing hardness. In recent work addressing this problem, Brynmor Chapman and the author~\cite{DBLP:conf/innovations/ChapmanW15} prove that $\NP \not\subset \P/\poly$ is equivalent to the existence of natural properties which are true of SAT but are useful against all polynomial-size ``SAT-solving'' circuits.

\section{Acknowledgments} I thank Amir Abboud, Russell Impagliazzo and Igor Carboni Oliveira, Steven Rudich, Rahul Santhanam, and the anonymous reviewers for useful comments and discussions. I also thank Emanuele Viola for a pointer to his paper with Eric Miles.

\bibliographystyle{alpha}
\bibliography{papers}

\end{document}